\documentclass[authoryear,12pt]{article}

\RequirePackage[OT1]{fontenc}

\RequirePackage[colorlinks,citecolor=blue,urlcolor=blue]{hyperref}
\usepackage[english]{babel}
\usepackage{amssymb,amsmath,amsthm}

\usepackage{times}
\usepackage{bm}
\usepackage{natbib}
\allowdisplaybreaks
\usepackage[plain,noend]{algorithm2e}

\newtheorem{theorem}{Theorem}

\makeatletter
\renewcommand{\algocf@captiontext}[2]{#1\algocf@typo. \AlCapFnt{}#2} 
\def\@algocf@capt@plain{top}
\renewcommand{\algocf@makecaption}[2]{%
  \addtolength{\hsize}{\algomargin}%
  \sbox\@tempboxa{\algocf@captiontext{#1}{#2}}%
  \ifdim\wd\@tempboxa >\hsize
    \hskip .5\algomargin%
    \parbox[t]{\hsize}{\algocf@captiontext{#1}{#2}}
  \else%
    \global\@minipagefalse%
    \hbox to\hsize{\box\@tempboxa}
  \fi%
  \addtolength{\hsize}{-\algomargin}%
}
\makeatother

\usepackage{graphicx}%
\usepackage{amsfonts}%
\usepackage{floatrow}
\usepackage{subcaption}
\usepackage[
singlelinecheck=false 
]{caption}
\usepackage{enumitem}
\usepackage{blindtext}
\usepackage{multirow}
\newfloatcommand{capbtabbox}{table}[][\FBwidth]

\allowbreak

\newcommand{\bbeta}{ \mbox{\boldmath $ \beta $} }

\newcommand{\R}{\mathbb{R}}
\renewcommand{\S}{\mathbb{S}}
\renewcommand{\ss}{\mathbf{s}}

\newcommand{\bone}{\textbf{1}}

\newcommand{\bB}{\textbf{B}}

\newcommand{\bE}{\textbf{E}}

\newcommand{\bF}{\textbf{F}}

\newcommand{\bh}{\textbf{h}}

\newcommand{\bs}{\textbf{s}}

\newcommand{\bw}{\textbf{w}}

\newcommand{\bx}{\textbf{x}}
\newcommand{\bX}{\textbf{X}}

\newcommand{\bY}{\textbf{Y}}

\title{Towards a Complete Picture of Stationary Covariance Functions on Spheres Cross Time}
\author{Philip A. White and Emilio Porcu}
\date{2019}

\begin{document}

\maketitle

\begin{abstract}
With the advent of wide-spread global and continental-scale spatiotemporal datasets, increased attention has been given to covariance functions on spheres over time.
This paper provides results for stationary covariance functions of random fields defined over $d$-dimensional spheres cross time. Specifically, we provide a bridge between the characterization in \cite{berg-porcu} for covariance functions on spheres cross time and Gneiting's lemma \citep{gneiting2002} that deals with planar surfaces.  

We then prove that there is a valid class of covariance functions similar in form to the Gneiting class of space-time covariance functions \citep{gneiting2002} that replaces the squared Euclidean distance with the great circle distance. Notably, the provided class is shown to be positive definite on every $d$-dimensional sphere cross time, while the Gneiting class is positive definite over $\R^d \times \R$ for fixed $d$ only. 

In this context, we illustrate the value of our adapted Gneiting class by comparing examples from this class to currently established nonseparable covariance classes using out-of-sample predictive criteria. These comparisons are carried out on two climate reanalysis datasets from the National Centers for Environmental Prediction and National Center for Atmospheric Research. For these datasets, we show that examples from our covariance class have better predictive performance than competing models.
\end{abstract}

\section{Introduction}\label{sec:intro}

In recent years, there has been a sharp increase in the prevalence of global or continental-scale spatiotemporal data due to satellite imaging, climate reanalyses, and wide-spread monitoring networks. Although Earth is not exactly spherical (it flattens at the pole), it is commonly believed that the Earth can be well approximated by a sphere \citep{gneiting2013, cas13}. With the goal of modeling data over large spatial scales, while accounting for the geometry of the Earth, there has recently been fervent research on modeling and inference for random fields on spheres as well as on spheres cross time. Recent examples provide a comprehensive overview of these topics, including \cite{gneiting2013}, \cite{jeong-jun}, \cite{Porcu-Bevilacqua-Genton}, \cite{berg-porcu}, and \cite{PAF}.

 Under Gaussianity, the covariance function is core to spatiotemporal modeling, inference, and prediction. Covariance functions are positive definite and showing that a candidate function is positive definite over spheres cross time
 often requires mathematical tools from harmonic analysis. Following the works of \cite{schoenberg} and \cite{gneiting2013} on spheres, the mathematical characterization of covariance functions on spheres cross time has been given by \cite{berg-porcu}. In addition, \cite{Porcu-Bevilacqua-Genton} provide examples of covariance functions for practitioners. As a special case of these covariance classes, some have adapted these classes for temporal models on circles (one-dimensional spheres) to account for seasonal patterns in temporal autocorrelation \citep[see][]{shirota2017,white2019}. Generalizations in the area of mathematical analysis include \cite{Guella3, Guella2, guella1} and \cite{Barbosa}.

For a random field on $\R^d \times \R$ with stationary covariance function $C: \R^d \times \R \to \R$, \cite{gneiting2002} showed the following characterization: if $C$ is continuous, bounded, symmetric, and integrable (over $\R^d$), then $C$ is a covariance function if and only if the function $C_{\boldsymbol{\omega}} :\R \to \R$, defined by
\begin{equation}
\label{criterion}   C_{\boldsymbol{\omega}}(u) = \int_{\R^d} {\rm e}^{\mathsf{i} \mathbf{h}^{\top}\boldsymbol{\omega}} C(\mathbf{h},u) {\rm d} \mathbf{h}, \qquad u \in \R, 
\end{equation}
is a covariance function for almost every $\boldsymbol{\omega} \in \R^d$. Here, $\mathsf{i}$ is the unit imaginary number and $\top$ is the transpose operator. This characterization has been the crux of many important results in spatiotemporal covariance modeling. Examples include the Gneiting class \citep{gneiting2002}, Schlather's generalized class \citep{schlather}, component-wise anisotropic covariances \citep{Porcu2006}, multivariate geostatistical modeling \citep{apa}, quasi-arithmetic construction \citep{porcu-mateu-christakos} and nonstationary models \citep{porcu-mateu-bevilacqua}. 
 For a given positive integer $d$, \eqref{criterion} proves the validity of the Gneiting class of covariance functions: 
\begin{equation}
\label{gneiting-euclidean} 
C(\mathbf{h} ;u) = \frac{\sigma^2}{\psi( u^2 )^{d/2} } \varphi \left( \frac{ \|\mathbf{h}\|^2 }{\psi( u^2)} \right), \qquad \mathbf{h} \in \R^d, u \in \R,
\end{equation}
where $\|\cdot\|$ is the Euclidean norm. The function $\varphi: [0,\infty) \to \R_+$ is completely monotonic; that is, $\varphi$ is infinitely differentiable on $(0,\infty)$, satisfying $(-1)^n \varphi^{(n)}(t) \ge 0$, $n \in \mathbb{N}$. Here, $\varphi^{(n)}$ denotes $n$th derivative and we use $\varphi^{(0)}$ for $\varphi$, where $\varphi(0)$ is required to be finite. The function $\psi: [0,\infty) \to \R_+$ is strictly positive with a completely monotonic derivative. Here and throughout, $\sigma^2$ is used to represent the spatiotemporal variance; that is, a scaling factor of a spatiotemporal correlation function. We also note that the function $C$ in (\ref{gneiting-euclidean}) is positive definite in $\R^d \times \R$ for a given positive integer $d$, but it is not positive definite on every $d$-dimensional Euclidean space cross time. 

Our paper focuses on two aspects of covariance modeling on $d$-dimensional spheres cross time. We first focus on Criterion \eqref{criterion} and its analogue on spheres over time. Our result provides an additional equivalence condition to those provided in \cite{berg-porcu}. We then provide an adaptation of the Gneiting class \eqref{gneiting-euclidean} to spherical domains, and show that it is positive definite over all $d$-dimensional spheres (including the Hilbert sphere) cross time. Further, our proof is based on direct construction, allows us to avoid Fourier inversion, and does not require a convergence argument that was originally used by \cite{gneiting2002}. \cite{Porcu-Bevilacqua-Genton} considered a variant of this problem, modifying the Gneiting class based on temporal rescaling of the spatial component. This idea was also suggested by \cite{gneiting2002}. In addition to a new Gneiting class for spheres over time, we adapt Heine's class of covariance functions \citep{heine}, originally proposed over two-dimensional Euclidean spaces, to $d$-dimensional spheres cross time. 

For estimation and prediction with the new covariance class, \eqref{eq:gneiting-nuovo}, we take a Bayesian approach using nearest neighbor Gaussian processes (NNGP) \citep{datta2016a,datta2016c}. Bayesian models allow for simple and rigorous uncertainty quantification through a single probabilistic framework that does not rely on asymptotic assumptions. Because Gaussian process (GP) models for large datasets, as we have with globally sampled spatiotemporal data, are often computationally intractable, we use the NNGP as a surrogate. Modeling with the NNGP enables scalable model fitting, inference, and prediction for real-data examples. Our discussion here adds to application areas for NNGPs as they have not been used for global data in the literature. 

For our data examples, we use daily near-surface temperature and cloud coverage from the first week of 2017 \citep{kalnay1996}. We only use the first week to keep computation times short. To be clear, we do not claim that covariance functions from our new covariance classes are preferable for all datasets. Indeed, for some datasets that we tested, but that we do not present here, the new covariance functions in this manuscript showed little or no predictive advantage. However, we highlight these datasets because they show that our new Gneiting class yields practical predictive benefits in some cases.

We start by giving background (Section \ref{Sec2}) for the theoretical results given in Section \ref{sec:theory}. In Section \ref{sec:theory}, we provide the analogue of \eqref{criterion} for covariance functions over spheres cross time. Additionally, we adapt the Gneiting class \eqref{gneiting-euclidean} to spheres cross time and show that, using a subclass of completely monotonic functions, the adapted Gneiting class can be used on all $d$-dimensional spheres cross time. Then, we provide an adaptation of the Heine covariance function, originally proposed in $\R \times \R$, to spheres cross time. Proofs of the theoretical results are technical and are deferred to the Appendix \ref{C}. 
We also provide a supplementary result in Appendix \ref{app:extra_thm} related to our main result in Section \ref{sec:theory}. We then turn our attention to modeling data using covariance functions from our adapted Gneiting class in Section \ref{sec:modeling}. In Section \ref{sec:Data}, we draw upon our modeling discussion for simulation studies and real data analyses. In our simulation studies, we explore parameter identifiability for examples from our adapted Gneiting covariance class and highlight some limitations. In our data examples, covariance functions from our adapted Gneiting covariance class have better out-of-sample predictive performance than covariance models currently in the literature, using mean absolute error, mean squared error, and continuous ranked probability scores as model comparison criteria. Finally, e provide concluding remarks in Section \ref{sec:conc}.

\section{Background}\label{Sec2}

%

Let $d$ be a positive integer. Here, we consider stationary Gaussian random fields on $d$-dimensional unit spheres $\S^d$ cross time (in $\R$), where $\S^d$ is defined to be $\{ \ss \in \mathbb{R}^{d+1} : \| \ss \| =  1 \}$. We use the unit sphere without loss of generality. These random fields are denoted $\{ Y(\ss,t), \, \ss \in \S^d, t \in \R \}$. We assume Gaussianity in modeling (Section \ref{sec:modeling}) which implies that finite dimensional distributions are completely specified by the mean and covariance function of the random field.

As a metric on $\S^d$, we use the great circle distance $\theta: \S^d \times \S^d \to [0,\pi]$, defined as the mapping
$$ (\ss_1,\ss_2) \mapsto \arccos \left ( \ss_1^{\top} \ss_2 \right ), \qquad \ss_1,\ss_2 \in \S^d.$$
We then consider covariance functions based on $\theta(\ss_1,\ss_2)$ and time difference $u = \mid t_1-t_2\mid$ ,
\begin{equation} \label{isotropy}
\text{cov} \left ( Y(\ss_1,t_1),Y(\ss_2,t_2) \right ) = C \left ( \theta(\ss_1,\ss_2), \mid t_1-t_2\mid \right ), \qquad (\ss_i,t_i) \in \S^d \times \R, 
\end{equation} 
where we take $\theta$ as an abbreviation for $ \theta(\ss_1,\ss_2)$. \cite{PAF} refer to such covariance functions as spatially geodesically isotropic and temporally symmetric, and \cite{berg-porcu} provide a mathematical characterization for these functions.

Covariance functions are positive definite, meaning that  for any collection $\{(\ss_k,t_k) \}_{k=1}^N \subset \S^d \times \R$ and constants $\{a_k \}_{k=1}^N$, $\sum_{k,h} a_k a_h C (\theta(\ss_k,\ss_h),\mid t_k-t_h\mid ) \ge 0$. It is worth noting that classes of positive definite functions on $\S^d \times \R$ are nested, meaning that positive definiteness on $\S^d \times \R$ implies positive definiteness on $\S^{d'} \times \R$ for $d'<d$, but the converse is not necessarily true. 

\cite{Porcu-Bevilacqua-Genton} proposed the {\em inverted} Gneiting class and define it as
\begin{equation}\label{eq:gneiting}
C(\theta ;u) = \frac{\sigma^2}{\psi_{[0,\pi]}( \theta )^{1/2} } \varphi \left( \frac{ u^2 }{\psi_{[0,\pi]}( \theta )} \right), \qquad \theta \in [0,\pi], u \in \R, 
\end{equation}
with $\varphi$ and $\psi$ as defined in \eqref{gneiting-euclidean}, and where $\psi_{[0,\pi]}$ denotes the restriction of $\psi$ to the interval $[0,\pi]$. 
In contrast to \eqref{gneiting-euclidean} which scales Euclidean distance by a function of the temporal lag, \eqref{eq:gneiting} rescales the temporal lag by a function of the great circle distance. This was also mentioned in \cite{gneiting2002}. It might be more intuitive to rescale space with time, as was done in \eqref{gneiting-euclidean}, proposing a structure like
\begin{equation}\label{eq:gneiting-nuovo}
C(\theta ;u) = \frac{\sigma^2}{\psi( u^2 )} \varphi \left( \frac{ \theta }{\psi( u^2 )} \right), \qquad \theta \in [0,\pi], u \in \R, 
\end{equation}
where, in this case, we do not need to restrict any of the functions $\varphi$ and $\psi$ to the interval $[0,\pi]$. Also, one might note that the function $\psi$ is not raised to the power $d/2$ as in \eqref{gneiting-euclidean}. Showing this construction is valid is nontrivial and receives an exposition in Section \ref{Sec3b}. 

One choice of $\varphi$, used to construct covariance functions in \eqref{gneiting-euclidean} and \eqref{eq:gneiting} is the Mat{\'e}rn class $\varphi(t)= {\cal M}_{\alpha,\nu}(t)$, $t \ge 0$, $\alpha,\nu>0$, defined as 
\begin{equation}
\label{matern} {\cal M}_{\alpha,\nu}(t) = \frac{2^{1-\nu}}{\Gamma(\nu)} \left ( \frac{t}{\alpha}\right )^{\nu} {\cal K}_{\nu } \left ( \frac{t}{\alpha}\right ), 
\end{equation}
where ${\cal K}_{\nu}$ is the MacDonald function \citep{grad}. One appeal of this class is the parameter $\nu$ that governs the smoothness at the origin \citep{stein-book}.   Unfortunately, ${\cal M}_{\alpha,\nu}(\theta)$ is not positive definite on $d$-dimensional spheres, unless $\nu \in (0,1/2]$ \citep{gneiting2013}, which makes this function less appealing to model spatial processes that are sufficiently smooth. 
 
\section{Theoretical Results}\label{sec:theory}

\subsection{The Generalized Gneiting Lemma on Spheres Cross Time} \label{Sec3a}


We begin our discussion with a criterion for covariance functions defined over $d$-dimensional spheres cross time. Let ${\cal G}_{k}^{\lambda}$ be the $k$th normalized Gegenbauer polynomial of order $\lambda>0$ \citep{dai-xu}. Gegenbauer polynomials form an orthonormal basis for the space of square-integrable functions ${\cal L}^2 ([0,\pi], \sin \theta^{d-1} {\rm d} \theta)$. 
 \begin{theorem} \label{partial-gneiting}
Let $d$ be a positive integer. Let $C:[0,\pi]\times \R \to \R$ be continuous and bounded with the $k$th related Gegenbauer transform, defined as
\begin{equation}
\label{sch-coeff} b_{k,d}(u) = \int_{0}^{\pi} C(\theta,u) {\cal G}_{k}^{(d-1)/2} (\cos \theta) \sin \theta^{d-1} {\rm d} \theta, \qquad u \in \R, k=0,1,\ldots,
\end{equation}
with $b_{k,d}: \R \to \R$, satisfying $\sum_{k=0}^{\infty} \int_{\R}\big|  b_{k,d}(u) \big|  {\rm d}u < \infty$.
Then, the following assertions are equivalent: 
\begin{enumerate}
\item[$1$.] $C(\theta,u)$ is the covariance function of a random field on $\S^d \times \R$;
\item[$2$.] the function $C_{\tau}:[0,\pi] \to \R$, defined as 
\begin{equation}
\label{partial-gneiting2} C_{\tau}(\theta) = \int_{-\infty}^{+\infty} {\rm e}^{- \mathsf{i} u \tau} C(\theta,u) {\rm d} u,
\end{equation}
is the covariance function of a random field on $\S^d$ for almost every $\tau \in \R$; 
\item[$3$.] for all $k=0,1,2,\ldots$, the functions $b_{k,d}: \R \to \R$, defined through \eqref{sch-coeff},
are continuous, positive definite on $\R$, and $\sum_{k} b_{k,d}(0) < \infty$.
\end{enumerate}
\end{theorem}

Some comments are in order. Equivalence of {$1$} and $3$ was shown by \cite{berg-porcu}. The result completes the picture that had been started by \cite{gneiting2013}, \cite{berg-porcu} and \cite{Porcu-Bevilacqua-Genton}. Equivalence of $1$ and $2$ provides the analogue of Gneiting's criterion in \eqref{criterion} for spheres cross time. Thus, Theorem \ref{partial-gneiting} gives insight into relationships between covariance functions on spheres and covariance functions on Euclidean spaces. In fact, application of Theorem \ref{partial-gneiting} provides a useful criterion (see Appendix \ref{app:extra_thm}) relating spatiotemporal covariances on $\R \times \R$ with covariance functions on $\S^3 \times \R$. 

The proof of Theorem \ref{partial-gneiting} is technical, and we defer it to Appendix \ref{C} to avoid mathematical obfuscation. 

\subsection{New Classes of Covariance Functions on Spheres Cross Time} \label{Sec3b}

We now detail our findings with two new classes of covariance functions on spheres over time. To do this, we need to introduce a new class of special functions. A function $\varphi:[0,\infty) \to \R$ is called a Stieltjes function if 
\begin{equation}
\label{stieltjes} \varphi(t) = \int_{[0,\infty)} \frac{{\mu} ({\rm d} \xi)}{t+\xi}, \qquad t \ge0,
\end{equation}
where $\mu$ is a positive and bounded measure. We require throughout $\varphi(0)=1$, which implies that 
$ \int \xi^{-1} \mu ({\rm d } \xi) =1$. Let us call ${\cal S}$ the set of Stieltjes functions. It has been proved that ${\cal S}$ is a convex cone \citep{berg2008}, with the inclusion relation ${\cal S} \subset {\cal C}$, where ${\cal C}$ is the set of completely monotone functions. The relation \eqref{stieltjes} shows that the function $t \mapsto (1+t)^{-1}$, $t \ge 0$, is a Stieltjes function. Using the fact that $\varphi \in {\cal S}$ if and only if  $1/\varphi$ is a completely Bernstein function \citep[for a definition, see ][]{porcu-schilling} we can get a wealth of examples, as the book by \cite{SSV} provides an entire catalogue of completely Bernstein functions.  We finally note that completely Bernstein functions are infinitely differentiable over $(0,\infty) $ and have a completely monotonic derivative. 

We are now able to state the following result.
\begin{theorem} \label{thm2}
Let $C: [0,\pi] \times \R $ be the mapping defined through \eqref{eq:gneiting-nuovo}, where $\varphi$ is a Stieltjes function on the positive real line, with $\varphi(0)=1$, and $\psi$ is strictly positive with a completely monotone derivative. Then, $C$ is a covariance function on $\S^d \times \R$ for all positive integers $d$.
\end{theorem}
Again, the proof is deferred to Appendix \ref{C}. This result completes the adaptation of the Gneiting class \citep{gneiting2002} to spheres cross time. Theorem \ref{thm2} allows $C$ in \eqref{eq:gneiting-nuovo} to be positive definite on every $d$-dimensional sphere under the condition that the function $\varphi$ is a Stieltjes function. As already mentioned, the class is rich, and there is a whole catalogue available from the book by \cite{SSV}. In addition, our proof in Appendix \ref{C} does not require any Fourier inversion techniques, nor convergence arguments as those used in \cite{gneiting2002}. 

We also note that the Mat{\'e}rn function ${\cal M}_{\alpha,\nu}$ cannot be used for the purposes of Theorem  \ref{thm2}. This is not surprising, as arguments in \cite{gneiting2013} show that the Mat{\'e}rn covariance function in \eqref{matern} can only be used in Theorem \ref{thm2} for $0 < \nu \le 1/2$. If one is interested in smoother realizations over spheres, then a common method involves using the Euclidean distance on spheres \citep{gneiting2013, PAF}, also called chordal distance, in \eqref{gneiting-euclidean}. In this case, any choice for $\nu>0$ preserves positive definiteness. At the same time, using chordal distance has a collection of drawbacks that have inspired constructive criticism in \cite{banerjee2005}, \cite{gneiting2013}, \cite{PAF} and \cite{alegria2-dimple}, to cite a few. We explore both possibilities and compare them in terms of predictive performance in Section \ref{sec:Data}.  

To introduce another class of covariance functions, we define the complementary error function ${\rm erfc}$ as $$ {\rm erfc}(t) = \frac{2}{\sqrt{\pi}} \int_{u}^{\infty} \exp \left ( - \xi^2 \right) {\rm d} \xi, \qquad t \ge 0, $$ and ${\rm erfc}(t)= 2- {\rm erfc}(-t)$ when $t$ is negative. We can show the following result.
\begin{theorem} \label{heine}
Let $\psi_{[0,\pi]}$ be the restriction to $[0,\pi]$ of a positive function with a completely monotonic derivative. Then, 
\begin{equation} \label{eq:heine} C(\theta,u)= {\rm e}^{-\mid u\mid }{\rm erfc } \left ( \sqrt{\psi_{[0,\pi]}(\theta)}- \frac{\mid u\mid }{2 \sqrt{\psi_{[0,\pi]}(\theta)}}\right ) + {\rm e}^{\mid u \mid }{\rm erfc } \left ( \sqrt{\psi_{[0,\pi]}(\theta)}+ \frac{\mid u\mid }{2 \sqrt{\psi_{[0,\pi]}(\theta)}}\right ),
\end{equation}
$\theta \in [0,\pi], u \in \R$, is a covariance function on $\S^d \times \R$ for all $d=1,2, \ldots$.
\end{theorem}
The class presented in \eqref{eq:heine} is related to a covariance class on $\R \times \R$ considered by \cite{heine}. Again, the proof is provided in Appendix \ref{C}.

\section{Modeling Nonseparable Spatiotemporal Data}\label{sec:modeling}

\subsection{Hierarchical Process Modeling for Spatiotemporal Data}


We illustrate the utility of one of our covariance classes (Theorem \ref{thm2}) using hierarchical NNGP models in a Bayesian setting. In spatial and spatiotemporal analyses, Bayesian models are often preferred for hierarchical modeling because they allow for simple and rigorous uncertainty quantification through a single probabilistic framework that does not rely on asymptotic assumptions \citep[see, e.g.,][]{gelman2014,banerjee2014,cressie2015}.

Spatiotemporal random effects for point-referenced data are often specified through a functional prior using a Gaussian process (GP). Gaussian processes are natural choices for modeling data that vary in space and time. However, likelihood computations for hierarchical GP models require inverting a square matrix with dimension equal to the size of the data, making GP models intractable in ``big-data'' settings. Many have addressed this computational bottleneck using either low-rank or sparse matrix methods \citep[see][for a review and comparison of some of these methods]{heaton2017}. 

Low-rank methods depend on selecting representative points, often called knots, that are used to approximate the original process \citep[see, e.g.,][]{higdon2002,banerjee2008,cressie2008,stein2008}. These models tend to oversmooth and often have poor predictive performance \citep[see][]{stein2014, heaton2017}.

In contrast to low-rank methods, inducing sparsity in either the covariance matrix or the precision matrix can reduce the computational burden. Covariance tapering creates sparsity in the covariance matrix by using compactly supported covariance functions \citep[see, e.g.,][]{furrer2006,kaufman2008}. These methods are generally effective for parameter estimation and interpolation; however, the allowable class of covariance functions is limited. On the other hand, inducing sparsity in the precision matrix has been leveraged to approximate GPs using Markov random fields \citep{Lindgren} or using conditional likelihoods \citep{vecchia1988,stein2004}. These approaches were extended to process modeling by \citet{gramacy2015} and \citet{datta2016a}. For discussion and further extension of these approaches, see \citet{katzfuss2017}. Unlike local approximate Gaussian processes \citep{gramacy2015}, the NNGP is itself a GP \citep{datta2016a} and has good predictive performance relative to other ``fast'' GP methods \citep[See][]{heaton2017}.

To specify an NNGP, we begin with a parent GP over $\mathbb{R}^d \times \mathbb{R}$ or $\S^d \times \mathbb{R}$. Nearest neighbor Gaussian process models induce sparsity in the precision matrix of the parent Gaussian model by assuming conditional independence given neighborhood sets constructed from directed acyclic graphs, yielding huge computational benefits \citep{datta2016a,datta2016c}. Modeling, model fitting, and prediction details for NNGP models are given in Appendix \ref{app:NNGP}.

\subsection{Examples of Covariance Functions}\label{sec:cov_examples}
\def\red{\textcolor{red}}
Here, we turn our attention to six nonseparable covariance functions used in simulation studies in Section \ref{app:sims} in our data analyses (Sections \ref{sec:air} and \ref{sec:cloud}). For all examples, we parameterize the models with variance $\sigma^2$ and use $c_s$ and $c_t$ to represent the strictly positive spatial and temporal scale parameters, respectively.

Explicitly, we consider two special cases from the Gneiting class \eqref{gneiting-euclidean} with $\varphi(t)={\cal M}_{c_s,\nu}(t)$, $t \ge 0$, for ${\cal M}_{c_s,\nu}$, defined in \eqref{matern}, obtained when $\nu=1/2$ and $\nu=3/2$,
\begin{equation}\label{eq:gneiting-chordal}
C(\bh ,u) = \frac{\sigma^2}{  \left(1 + \left(\frac{\mid u\mid }{c_t} \right)^\alpha \right)^{\delta + \beta d/2}} {\cal M}_{c_s,\nu} \left ( \frac{\|\bh\|}{\left(1 + \left(\frac{\mid u\mid}{c_t} \right)^\alpha\right)^{\beta/2} }\right ), \quad (\bh,u) \in \R^d \times \R. \end{equation}
These choices correspond to ${\cal M}_{c_s,1/2}(t)= \exp(-t/c_s)$ and ${\cal M}_{c_s,3/2}(t)= \exp(-t/c_s) (1+t/c_s)$, $t \ge0$. Here, we work on the sphere, thus $\| \cdot \|$ refers to chordal distance. The parameters restriction is $\delta >0$, $\beta \in (0,1]$ and $\alpha \in (0,2]$. 
%
%
%
%
%
%

Next, we consider a pair of similar covariance models from the inverted Gneiting class  
\citep{Porcu-Bevilacqua-Genton} in \eqref{eq:gneiting}):
\begin{equation}\label{eq:invgneiting1}
C(\theta,u) = \frac{\sigma^2}{\left( 1 + \left(\frac{\theta}{c_s}\right)^\alpha \right)^{\delta + \beta/2} } \exp\left( - \frac{ | u| ^{2\gamma} }{c_t^{2 \gamma} \left( 1 + \left(\frac{\theta}{c_s}\right)^\alpha \right)^{\beta \gamma} } \right), \qquad (\theta,u) \in [0,\pi] \times \R,
\end{equation}
where $\delta>0$, and where  $\beta, \alpha$ and $\gamma$ belong to the interval $(0,1]$. The second we consider uses the generalized Cauchy covariance function \citep{GS2004} for the temporal margin,
that is $\psi(u) = (1 + (|u|/c_s)^{\gamma} )^{-\lambda}$:
\begin{equation}\label{eq:invgneiting2}
C(\theta,u) = \frac{\sigma^2}{\left( 1 + \left(\frac{\theta}{c_s}\right)^\alpha \right)^{\delta + \beta/2} } \left( 1+ \frac{ |u| ^{2\gamma} }{c_t^{2\gamma} \left( 1 + \left(\frac{\theta}{c_s}\right)^\alpha \right)^{\beta \gamma} } \right)^{-\lambda},  \qquad (\theta,u) \in [0,\pi] \times \R,
\end{equation}
with $\delta,\lambda >0$, and where  $\beta, \alpha$ and $\gamma$ belong to the interval $(0,1]$. 


As a first example from our new adapted Gneiting class on spheres cross time (see Theorem \ref{thm2}), we chose the Stieltjes function $$\varphi(t)= \kappa \frac{1-{\rm e}^{-2 \sqrt{t+1}}}{\sqrt{t+1}}, \qquad t \ge 0, $$ with $\kappa:= 1/(1-{\rm e}^{-2})$ a normalization constant. 
We then pick the function $\psi(t)= (1+t^{\alpha})^{\delta}$, that is a Bernstein function for $\alpha,\delta \in (0,1]$. Thus, we have 
\begin{equation}\label{eq:our_mod1}
C(\theta,u) = \frac{\sigma^2 \kappa}{\left( 1 + \left ( \frac{\mid u\mid }{c_t} \right)^\alpha \right)^{\delta}}  
\frac{1-{\rm e}^{-2 \left (1 + \frac{\theta}{c_s \left  ( 1+ \left ( \frac{|u|}{c_t} \right )^{\alpha} \right )^{\delta}} \right )^{1/2}}}{\left (1 + \frac{\theta}{c_s \left ( 1+ \left ( \frac{|u|}{c_t} \right )^{\alpha} \right )^{\delta}} \right )^{1/2}}.
\end{equation} 
For the second, we again propose a generalized Cauchy covariance function for the spatial margin, obtaining
\begin{equation}\label{eq:our_mod2}
C(\theta,u) = \frac{\sigma^2}{\left( 1 + \left(\frac{\mid u\mid }{c_t}\right)^\alpha \right)^{\delta + \beta /2} } \left( 1+ \frac{ \theta^{\gamma} }{c_s^{\gamma}  \left( 1 + \left(\frac{\mid u\mid }{c_t}\right)^\alpha \right)^{\beta \gamma} } \right)^{-\lambda},  \qquad (\theta,u) \in [0,\pi] \times \R,
\end{equation}
where $\delta >0$, $\beta, \gamma \in (0,1]$, $\alpha \in (0,2]$ and $\lambda >0$.

For all models in our simulation study and our data analyses, we include an independent Gaussian error term with variance $\tau^2$ in the model. The variance $\tau^2$ is often called a nugget and accounts for potential discontinuities at the origin of the covariance function. In other words, the nugget accounts for sources of uncertainty that are not explained or captured by our spatiotemporal model. 

\subsection{Model Comparison}\label{sec:mod_comp}

To compare models that differ in terms of covariance specification, we propose the following criteria for comparing predictions to hold-out data $y_i$: 90\% predictive interval coverage, predictive mean square or absolute error (defined as $( \bE(Y_i\mid \bY_{obs}) - y_i)^2$ or $\mid  \bE(Y_i\mid \bY_{obs}) - y_i\mid $, respectively), where $\bY_{obs}$ denotes observations.
Besides these common criteria, we also use a strictly proper scoring rule \citep{gneiting2007}, the continuous ranked probability score (CRPS), defined as
\begin{equation}
\text{CRPS}(F_i,y_i) = \int^\infty_{-\infty} (F_i(x) -  \bone(x \geq y_i) )^2 {\rm d} x =  \bE\mid Y_i - y_i \mid  - \frac{1}{2}\bE \mid  Y_i - Y'_i \mid,
\end{equation}
where $Y_i$ and $Y_i'$ follow the predictive distribution $F_{i}$ \citep[see][for early discussion on CRPS]{brown1974,matheson1976}.. An empirical estimate of the continuous ranked probability score, using $M$ posterior predictive samples $Y_{i,1},...,Y_{i,M}$ from $Y_{i} \mid  \bY_{obs}$, is
\begin{equation}
\text{CRPS}(\hat{F}_i,y_i) = \frac{1}{M} \sum_{j=1}^M \mid Y_{i,j} - y_i\mid  - \frac{1}{2M^2} \sum_{j=1}^M \sum_{k=1}^M \mid  Y_{i,j} - Y_{i,k} \mid.
\end{equation}
We average continuous ranked probability scores over all hold-out data to obtain a single value for comparison.

\section{Simulation and Data Examples}\label{sec:Data}

Using only covariance functions \eqref{eq:our_mod1} and \eqref{eq:our_mod2}, we provide a brief simulation study in Section \ref{app:sims} to explore the identifiability of covariance model parameters using an NNGP model. In Sections \ref{sec:air} and \ref{sec:cloud}, we illustrate practical predictive advantages of the new Gneiting class using spherical distance (Theorem \ref{thm2}) using two climate reanalysis datasets from the National Centers for Environmental Prediction and National Center for Atmospheric Research \citep{kalnay1996}. For both analyses, we use the first week of the 2017 dataset.

\subsection{Simulation Study}\label{app:sims}


Here, we present simulation studies to explore the empirical identifiability of covariance model parameters.
To do this, we simulate many datasets that differ in their covariance specification, using either (\ref{eq:our_mod1}) and (\ref{eq:our_mod2}). For each covariance function, we fix parameters and generate 1,000 datasets from the following generative model: 
\begin{align}
Y(\bs,t) &= w(\bs,t) + \epsilon(\bs,t), \label{eq:gen_mod} \\
w(\bs,t) &\sim \text{GP}(0,C((\bs,t),(\bs',t'))),  \nonumber \\
\epsilon(\bs,t) &\sim \text{GP}(0,\tau^2 \delta^\bs_\bs \delta^t_t),
 \nonumber
\end{align}  
where observations lie on an evenly spaced grid of latitudes ranging from -90 to -60 (5$^\circ$ spacing) and longitudes between -180 and 0 (5$^\circ$ spacing). This grid is repeated from times 1 to 10, giving $N=4330$. While we use a full GP specification for simulation and the dataset is not particularly large when fitting a single model, we fit these data using a hierarchical NNGP model with $m=25$ neighbors because this mirrors modeling approach in Section \ref{sec:modeling} and because we fit 1,000 simulated datasets per simulation (6,000 in total). Neighbors are selected using simple rectangular neighborhood sets using great-circle (spherical) distance to define nearness \citep[see][]{datta2016c}. 

When simulating from \eqref{eq:gen_mod} using (\ref{eq:our_mod1}) as the covariance function, we use with $\sigma^2 \kappa = 4$, $c_s = 0.2$, $c_t = 2$, $\alpha = 1/2$, $\delta = 1/2$, and $\tau^2 = 1$ for the noise term $\epsilon(\bs,t)$. For $\tau^2$ and $\sigma^2$, we use inverse-gamma priors with 0.1 as both the shape and scale (corresponding to the rate of a gamma distribution) parameters. We assume $c_s \sim \text{Unif}(0,\pi)$, $c_t \sim \text{Unif}(0,10)$, $\alpha \sim \text{Unif}(0,1]$, and $\delta \sim \text{Unif}(0,1]$, \emph{a priori}.

When using (\ref{eq:our_mod2}) in the generative model \eqref{eq:gen_mod}, we set $\sigma^2 = 4$, $c_s = 0.2$, $c_t = 2$, $\alpha =1$, $\beta = 1/2$, $\delta = 3/4$, $\lambda = 1$, $\gamma = 1/2$, and $\tau^2 = 1$ for the noise term $\epsilon(\bs,t)$. For simplicity, we constrain $\delta + \beta/2 = 1$ , $\gamma = 1/2$, and $\lambda = 1$. As before, we use inverse-gamma priors with 0.1 as both the shape and scale parameters for $\tau^2$ and $\sigma^2$. As before, we assume $c_s \sim \text{Unif}(0,\pi)$, $c_t \sim \text{Unif}(0,10)$, $\alpha \sim \text{Unif}(0,2]$, $\beta \sim \text{Unif}(0,1]$, and $\delta \sim \text{Unif}(0,1]$, \emph{a priori}. 

We explore the effect of fixing some model parameters to examine how model identifiability is affected. Specifically, we consider fixing combinations of $c_s$, $c_t$, and $\alpha$ in (\ref{eq:our_mod1}) and (\ref{eq:our_mod2}) to the true value. In this setting, we re-fit the models described above, keeping all other specifications the same as described. We present the 90\% empirical coverage rates for all parameters in Table \ref{tab:coverage}. In this table, dashes indicate that parameters are fixed. For (\ref{eq:our_mod1}), we see improved coverage rates (i.e., closer to 90\%) for $\sigma^2$ and $\alpha$ when range parameters $c_s$ and $c_t$ are fixed; however, $\delta$ shows significant under coverage when range parameters are fixed. The results for (\ref{eq:our_mod2}) are similar. When range parameters are fixed, coverage rates for $\sigma^2$ and $\beta$ are closer to 90\%. For this covariance model, $\alpha$ has slightly worse coverage when range parameters are fixed.

\begin{table}[H]
\centering
\begin{tabular}{rrrrrrr}
  \hline
 Model & $\sigma^2$ & $\tau^2$ & $c_s$ & $c_t$ & $\alpha$ & $\delta$ or $\beta$ \\ 
  \hline
(\ref{eq:our_mod1}) & 0.72 & 0.87 & 0.71 & 0.84 & 0.74 & 0.81 \\ 
(\ref{eq:our_mod1}) & 0.79 & 0.87 & ---- & ---- & 0.82 & 0.62 \\ 
(\ref{eq:our_mod1}) & 0.80 & 0.89 & ---- & ---- & ---- & 0.55 \\ 
(\ref{eq:our_mod2}) & 0.73 & 0.88 & 0.67 & 0.89 & 0.75 & 0.99 \\ 
(\ref{eq:our_mod2}) & 0.82 & 0.89 & ---- & ---- & 0.69 & 0.92 \\ 
(\ref{eq:our_mod2}) & 0.82 & 0.90 & ---- & ----  & ---- & 0.92 \\ 
   \hline
\end{tabular}
\caption{90\% empirical coverage rates. Here, we use dashes when parameters are fixed on the parameters used in simulating the data.}\label{tab:coverage}
\end{table}

These simulation studies highlight some of the limitations in estimating parameters of covariance functions from Theorem \ref{thm2}.
For both (\ref{eq:our_mod1}) and (\ref{eq:our_mod2}), we note that there is limited parameter identifiability, particularly for spatial range parameter $c_s$ and spatiotemporal variance $\sigma^2$. Although we do not provide identifiability proofs, this result seems similar to the limited identifiability of the Mat\'ern class that is identifiable up to $\sigma^2 {c_s}^{-2\nu}$ \citep[see][]{zhang2004}. Apparently, this issue becomes even more complex under the space-time setting, and the lack of theoretical results for space-time asymptotics and equivalence of Gaussian measures make this problem very difficult. In this simulation, we find improved parameter identifiability when we fix spatial range parameters. We emphasize that the primary goal of our analyses is comparing predictive performance. However, if the unbiased estimation of covariance parameters is the primary goal, then a multi-stage fitting process can improve estimation \citep{mardia1984}.

\subsection{Surface Air Temperature Reanalysis Data}\label{sec:air}

For this section, we utilize the 2017 National Centers for Environmental Prediction/National Center for Atmospheric Research daily average 0.995 sigma level (near-surface) temperature reanalysis data \citep{kalnay1996}. Air temperature at 0.995 sigma level is defined to be the temperature taken at an air pressure 0.995 $\times$ surface air pressure. 

The foundations of global temperature change are well established \citep[see, e.g.,][]{folland2001,hansen2006,hansen2010}. Furthermore, air temperature changes have, along with other changes in climate, a wide and deep impact on global biological systems \citep[see][]{parmesan2003,thomas2004,held2006}. For these reasons, many climate models are dedicated to understanding past and predicting future temperature changes \citep[see][for some discussion and comparisons about the various analyses of surface air temperature]{simmons2004}. 

The daily near-surface temperature reanalysis data represent daily temperature averages over a global grid with $2.5^{\circ}$ spacing for latitude and longitude. We thin the data to $5^{\circ}$ spacing for latitude and longitude to carry out a model comparison on the hold-out data. In this dataset, we have observations at 2522 unique spatial locations, giving 17654 total observations. The averages of near-surface temperature over the first week of January are plotted in Figure \ref{fig:temp_map}. Additionally, we give the density estimate of near-surface temperature for each day in Figure \ref{fig:temp_hist}. Figure \ref{fig:temp_hist} shows that the overall temperature distribution is similar across days, and Figure \ref{fig:temp_map} demonstrates a clear spatial structure. Because our covariance model allows space to be scaled by time while using the spherical distance, we expect that our model may be able to capture the strong spatial structure in this data more effectively than the models with which we compare them.
\begin{figure}[H]
\begin{center}
   \begin{subfigure}[b]{.59\textwidth}
\includegraphics[width=\textwidth]{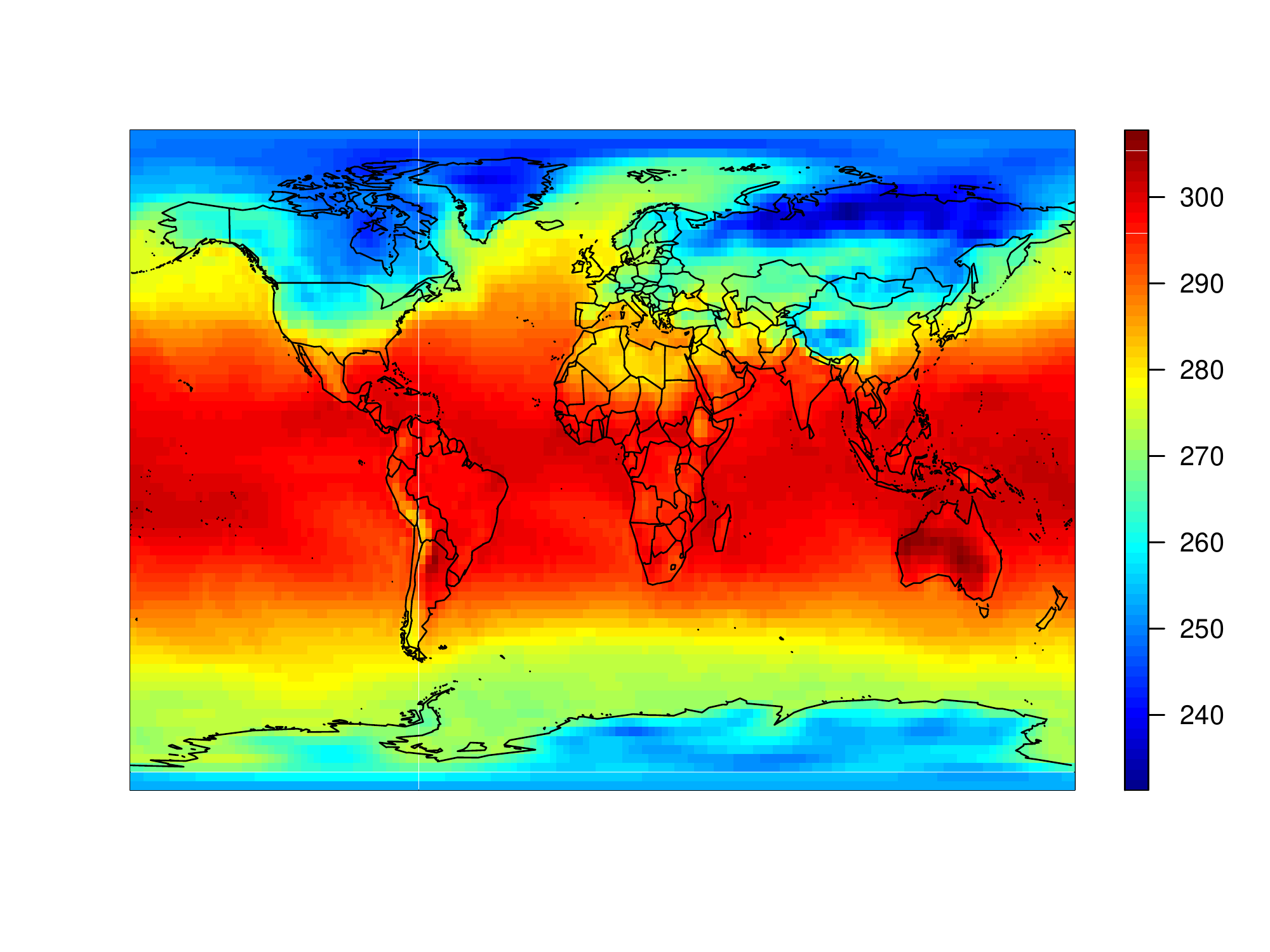}
\vspace{-10mm}
      \subcaption{ Near-surface temperature heat map in degrees Kelvin (K).}\label{fig:temp_map}
   \end{subfigure}
      \begin{subfigure}[b]{.4\textwidth}
\includegraphics[width=\textwidth]{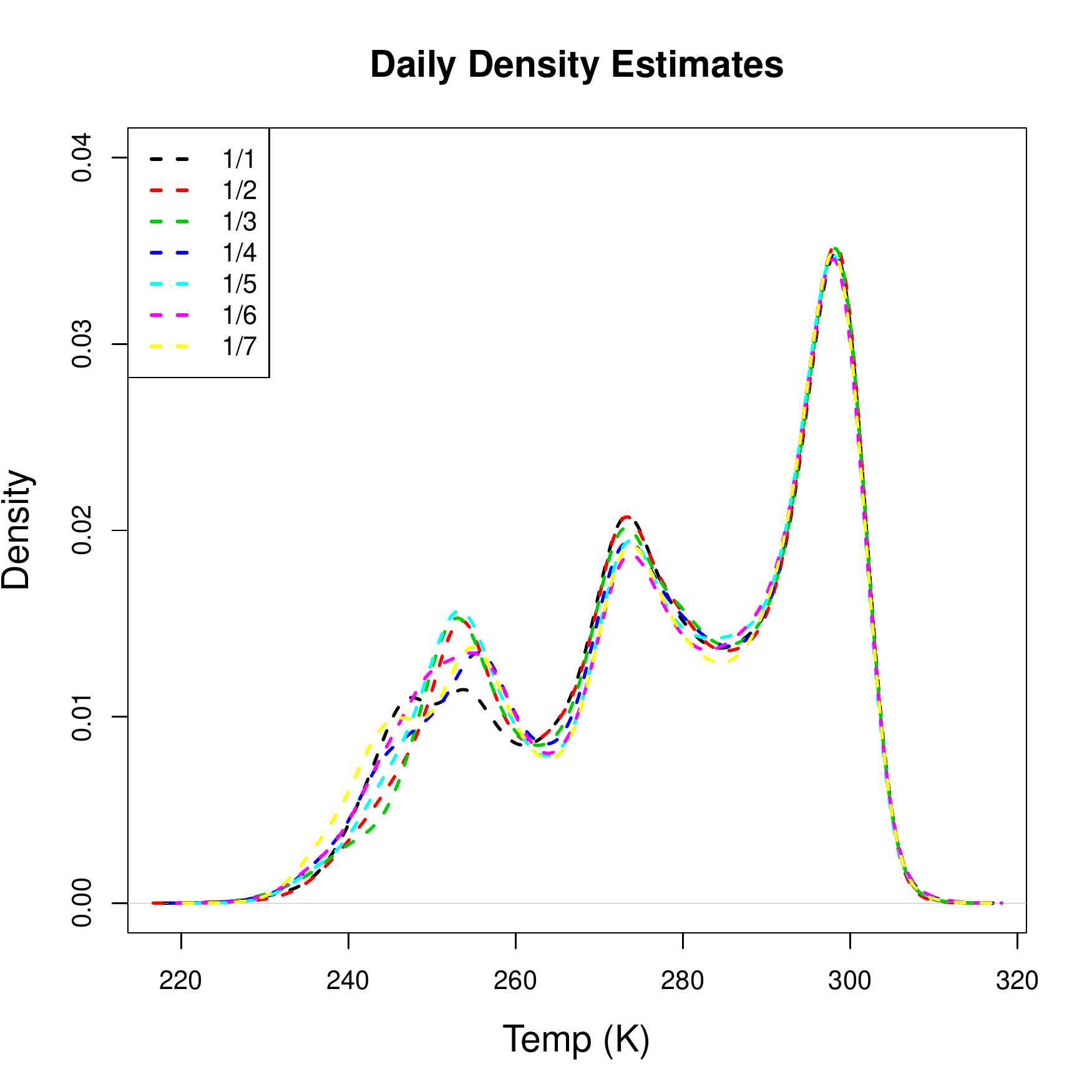}
      \subcaption{ Kernel density estimates of global near-surface temperature for each day.}\label{fig:temp_hist}
   \end{subfigure}
\end{center}
\end{figure}
With 17654 data, carrying out fully Bayesian inference using a full GP model is computationally burdensome; thus, we utilize a NNGP model. For these models, we use simple neighborhood selection presented in \citet{datta2016c} using $m = 25$ neighbors, using the five nearest neighbors at the five most recent times, including the current time. We utilize two covariance models from each of the Gneiting class, the inverted-Gneiting class using spherical distance, and our new Gneiting class (Theorem \ref{thm2}). These models are presented in \eqref{eq:gneiting-chordal}
to \eqref{eq:our_mod2}. For all models, we use inverse-gamma prior distributions for $\tau^2$ and $\sigma^2$ with 0.1 for both the shape and scale parameters (corresponding to the rate parameter of a gamma distribution). We use prior distributions of $c_s \sim \text{Unif}(0,\pi)$, $c_t \sim \text{Unif}(0,10)$, $\alpha \sim \text{Unif}(0,2]$, and $\beta \sim \text{Unif}(0,1]$ for correlation function parameters. Because many covariance models have limited parameter identifiability, we constrain $\delta + \beta d/2 = 1$ for \eqref{eq:gneiting-chordal}, and $\delta + \beta/2 = 1$ for \eqref{eq:invgneiting1}, \eqref{eq:invgneiting2}, and \eqref{eq:our_mod2}.

We compare these six models in terms of predictive performance on a randomly selected subset of the hold-out data. In total, we use 1000 locations over the week, giving 7000 hold-out observations. These hold-out locations are plotted in Figure \ref{fig:air_holdout} in Appendix \ref{app:holdout}. We compare these models in terms of predictive root mean squared error, mean absolute error, continuous ranked probability score, and 90\% prediction interval coverage, as discussed in Section \ref{sec:mod_comp}. For predictions, neighbors are chosen to make prospective predictions using 25 neighbors (See Appendix \ref{app:NNGP} for details on modeling and prediction). 

The results presented are based on 25,000 posterior draws after a burn-in of 5,000 iterations using a Gibbs sampler presented in Appendix \ref{app:NNGP}. These posterior samples are used for prediction and posterior inference. The results of the model comparison are given in Table \ref{tab:mod_comp1}. For this data, the covariance models from our class \eqref{eq:our_mod1} and \eqref{eq:our_mod2} had the best out-of-sample predictive performance, and the model \eqref{eq:our_mod1} was the very best. For comparison, the Gneiting class using chordal distance had continuous ranked probability scores 7\% and 16\% higher than the best model for $\nu = 1/2$ and $\nu = 3/2$, respectively. Both models from the inverted Gneiting class with spherical distance were 13\% worse in terms of continuous ranked probability scores.
\begin{table}[H]
\centering
\begin{tabular}{lrrrrr}
  \hline
 Equation & PRMSE & PMAE & 90\% Coverage & CRPS & Relative CRPS \\ 
  \hline
 \eqref{eq:gneiting-chordal} and $\nu=1/2$ & 6.44 & 4.58 & 0.91  & 3.50 & 1.07 \\
  \eqref{eq:gneiting-chordal} and $\nu=3/2$  & 7.15 & 5.08 & 0.91  & 3.81 & 1.16 \\
 \eqref{eq:invgneiting1} & 6.79 & 4.85 & \textbf{0.90}  & 3.70 & 1.13 \\
 \eqref{eq:invgneiting2}  & 6.78 & 4.84 & \textbf{0.90}  & 3.69 & 1.13 \\
 \eqref{eq:our_mod1} & \textbf{6.02} & \textbf{4.26} & \textbf{0.90}  & \textbf{3.28} & \textbf{1.00} \\ 
  \eqref{eq:our_mod2} & \textbf{6.04} & \textbf{4.40} & \textbf{0.90}  & \textbf{3.35} & \textbf{1.02}\\ 
   \hline
\end{tabular}
\caption{Predictive performance of competing covariance models for the cloud cover dataset. For brevity in the table, let PRMSE and PMAE denote predictive mean squared error and mean absolute error, respectively.  Relative CRPS is scaled such that the lowest is one. Bolded entries are used to indicate best model performances, i.e. lowest PRMSE, PMAE, and CRPS and 90\% interval coverage closest to 90\%. }\label{tab:mod_comp1}
\end{table}
For the best model, \eqref{eq:our_mod1}, we provide posterior summaries in Table \ref{tab:post_sum1}. Additionally, we display the correlation contour as function of spherical distance $\theta$ and time $t$ for the posterior mean in Figure \ref{fig:air_contour}. Correlation is very persistent as a function of time; thus, decreases in autocorrelation are almost completely attributable to changes in spatial location. 


\begin{table}[H]
\centering
\begin{tabular}{rrrrr}
  \hline
  & Mean & Standard Deviation & 2.5\%  & 97.5\%  \\ 
  \hline
    $\tau^2$  & 21.140 & 0.553 & 20.088 & 22.250 \\ 
 $\sigma^2$  & 105.897 & 4.253 & 97.953  & 114.498 \\ 
 $c_s$ &  0.994 & 0.024 & 0.952 & 1.025 \\
 $c_t$  & 2.783 & 3.243 & 0.026 & 9.624 \\  
$\alpha$  & 0.017 & 0.020 & 0.001 & 0.067 \\ 
$\delta$  & 0.090 & 1.54e-3 & 0.088 & 0.092 \\ 
   \hline
\end{tabular}
\caption{Posterior summaries for the near-surface air temperature dataset for parameters for the model fit using (\ref{eq:our_mod1}). Percentiles (2.5\% and 97.5\%) represent posterior percentiles.}\label{tab:post_sum1}
\end{table}
\begin{figure}[H]
\begin{center}
\includegraphics[width=.48\textwidth]{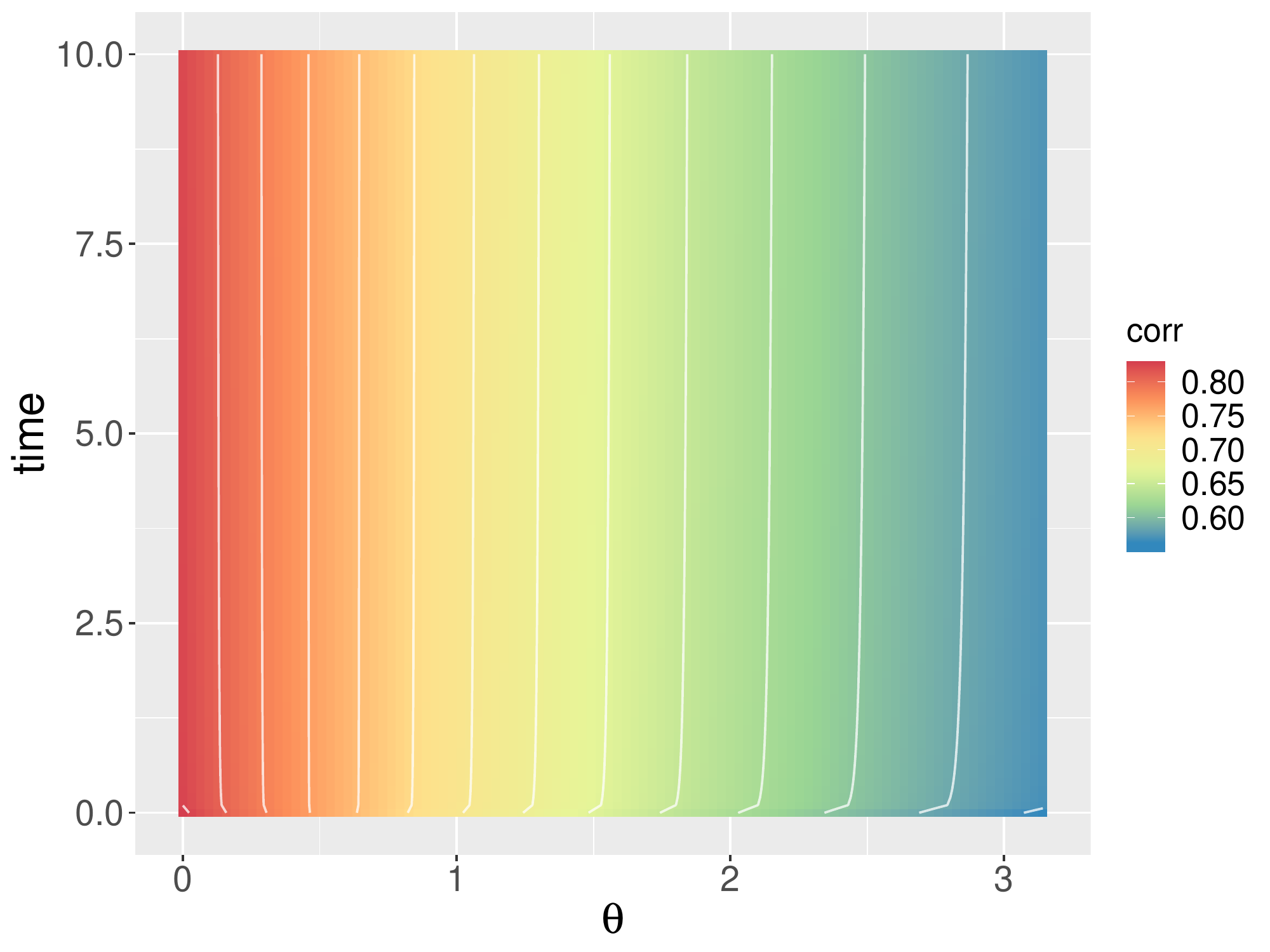}
\vspace{-5mm}
\caption{Posterior mean correlation contour plot for the near-surface air temperature.}\label{fig:air_contour}
\end{center}
\end{figure}
We can also compute posterior summaries for $\sigma^2/ (\sigma^2 + \tau^2)$, interpreted as the proportion of total variance attributable to the spatiotemporal random effect. Here, the posterior mean of $\sigma^2/ (\sigma^2 + \tau^2)$ is 0.834. In other words, our selected covariance model accounts for 83.4\% of the total variance. In Table \ref{tab:post_sum1} and Figure \ref{fig:air_contour}, $c_s$ and $c_t$ suggest that the surface temperature process exhibits persistent correlation over both space and time. Low values of $\alpha$ in \eqref{eq:our_mod1} add temporal smoothness.

\subsection{Total Cloud Coverage Reanalysis Data}\label{sec:cloud}

For this section, we utilize the 2017 National Centers for Environmental Prediction/National Center for Atmospheric Research daily average total cloud coverage reanalysis data \citep{kalnay1996}. Total cloud coverage is defined as the fraction of the sky covered by any visible clouds, regardless of type. Total cloud coverage takes values between 0 and 100, representing a percentage of cloud coverage. Values of total cloud coverage close to 0 indicate clear skies, values from 40 to 70 percent represent broken cloud cover, and overcast skies correspond with 70 to 100 percent. 

The degree of cloudiness impacts how much solar energy radiates to the Earth \citep[see, e.g.,][]{svensmark1997}. Total cloud coverage, like changes in global surface temperature, has been impacted by global climate changes \citep[see, e.g.,][]{melillo1993,wylie2005}, and changes in total cloud coverage are linked with many biological changes \citep[see][]{pounds1999}. Thus, tracking, predicting, and anticipating changes in cloudiness have important implications for understanding global climate changes and their effects on ecosystems. 

The daily total cloud coverage reanalysis data represent daily averages and are given on a global grid with $1.9^{\circ}$ spacing for latitude and $1.875^\circ$ spacing for longitude. The spatial averages of cloud coverage over the first week of January are plotted in Figure \ref{fig:cloud_map}. This map shows clear spatial variability that suggests that a spatial model is appropriate. We provide density estimates of total cloud coverage for each day of the week in Figure \ref{fig:cloud_hist}. As we did with temperature, we see from Figure \ref{fig:cloud_hist} that cloud coverage appears similar across days.
\begin{figure}[H]
\begin{center}
   \begin{subfigure}[b]{.58\textwidth}
\includegraphics[width=\textwidth]{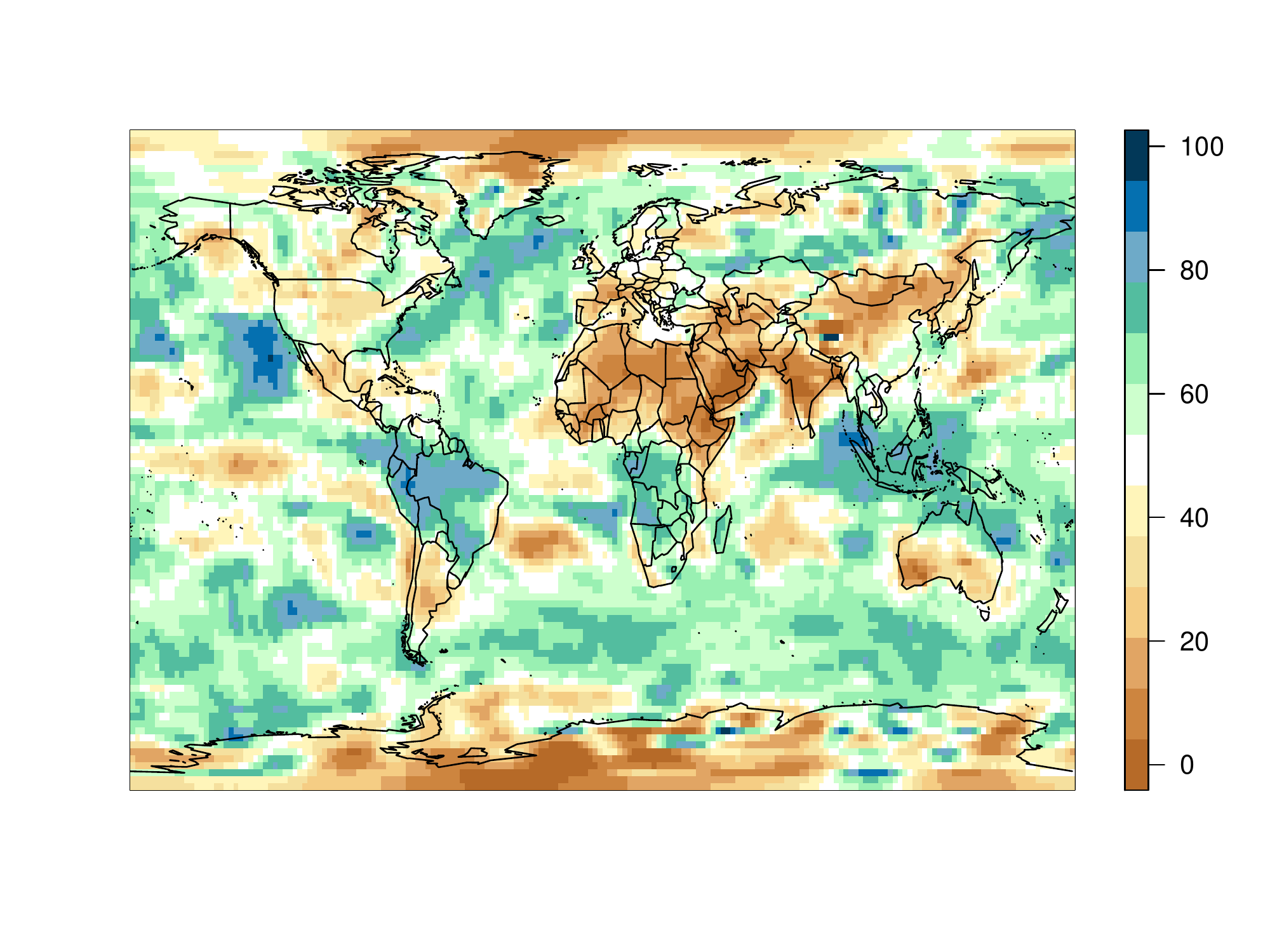}
\vspace{-10mm}
      \subcaption{ Total cloud coverage heat map in percentage, taking values 0 to 100.}\label{fig:cloud_map}
   \end{subfigure}
      \begin{subfigure}[b]{.4\textwidth}
\includegraphics[width=\textwidth]{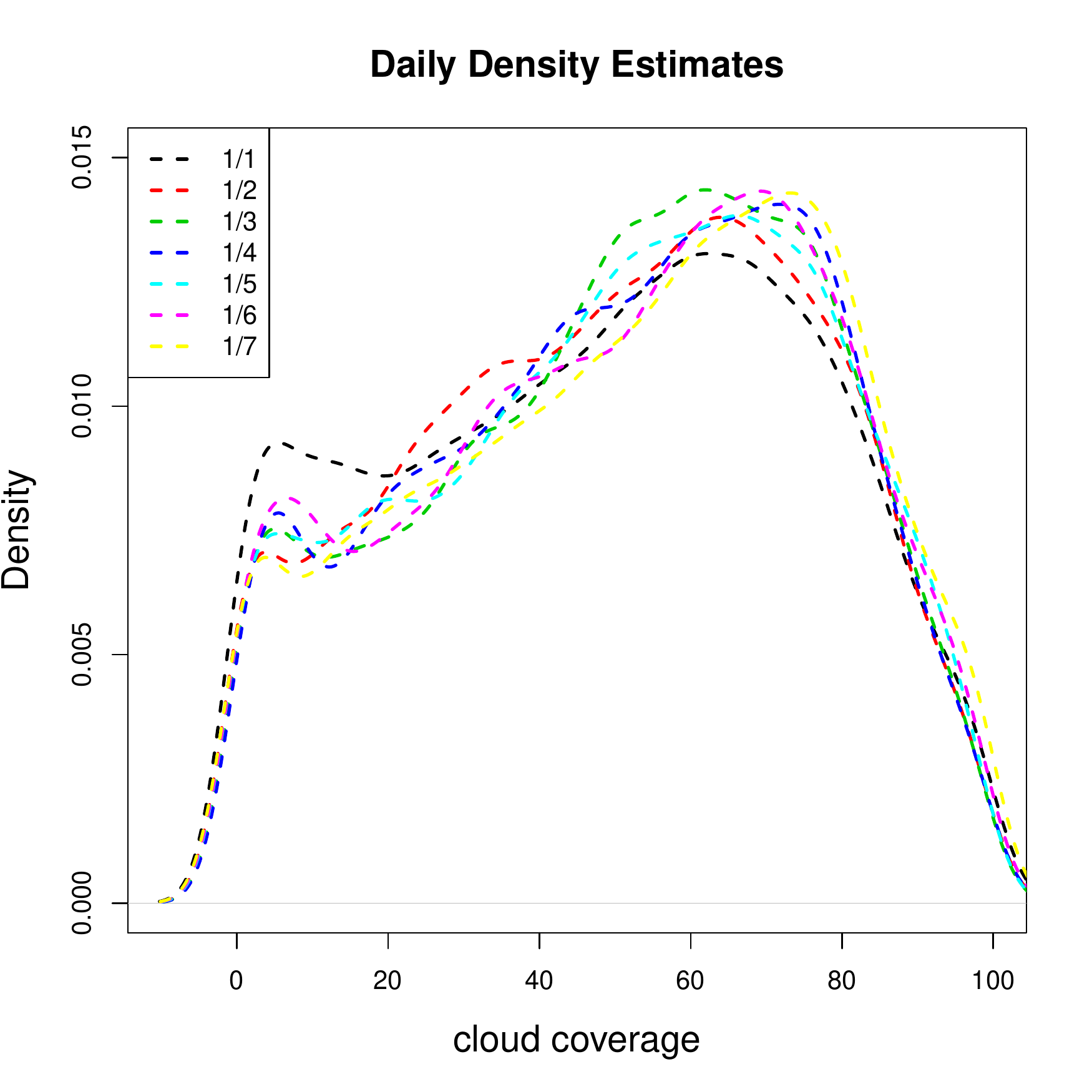}
      \subcaption{Kernel density estimates of total cloud coverage for each day.}\label{fig:cloud_hist}
   \end{subfigure}
\end{center}
\end{figure}

Again, we thin the data to $3.8^{\circ}$ spacing for latitude and $3.75^{\circ}$ for longitude to carry out model comparison on hold-out data. In total, we have 4512 unique spatial locations, giving 31584 total observations. With 31584 data, carrying out fully Bayesian inference using a full Gaussian process model is intractable; thus, we utilize a nearest neighbor Gaussian process model. For these models, we use the same neighborhood formulations and fit the same covariance models with the same prior distribution to these data as we did in Section \ref{sec:air}. 

To obtain a test set to compare the six competing models, we randomly select 1000 locations and predict at these locations over the time-span of our data, giving a test set of size 7000. These hold-out locations are plotted in Figure \ref{fig:cloud_holdout} in Appendix \ref{app:holdout}. As before, we compare these models in terms of predictive mean squared error, mean absolute error, continuous ranked probability scores, and 90\% prediction interval coverage, as discussed in Section \ref{sec:mod_comp}. The results of the model comparison are given in Table \ref{tab:mod_comp2}. 

\begin{table}[H]
\centering
\begin{tabular}{lrrrrr}
  \hline
 Equation & PRMSE & PMAE & 90\% Coverage & CRPS & Relative CRPS \\ 
  \hline

(\ref{eq:gneiting-chordal}) and $\nu=1/2$ & 13.92 & 11.08 & 0.95  & 7.85 & 1.32 \\
(\ref{eq:gneiting-chordal}) and $\nu=3/2$  & 13.35 & 10.51 & 0.94  & 7.53 & 1.26\\
 (\ref{eq:invgneiting1})  & 19.89 & 8.27 & \textbf{0.90}  & 8.27 & 1.39 \\
 (\ref{eq:invgneiting2})  & 19.31 & 8.18 & \textbf{0.90}  & 8.15 & 1.37 \\
 (\ref{eq:our_mod1})  & 13.24 & 10.41 & 0.95  & 7.51 & 1.26 \\ 
(\ref{eq:our_mod2}) & \textbf{10.68} & \textbf{7.66} & 0.95  & \textbf{5.96} & \textbf{1.00} \\ 
   \hline
\end{tabular}
\caption{Predictive performance of competing covariance models for the total cloud coverage dataset. For brevity in the table, let PRMSE and PMAE denote predictive mean squared error and mean absolute error, respectively. Relative CRPS is scaled such that the lowest is one. Bolded entries are used to indicate best model performance, i.e. lowest PRMSE, PMAE, and CRPS and 90\% interval coverage closest to 90\%. } \label{tab:mod_comp2}
\end{table}

Again, an example from our new class was best and models in terms of prediction for the total cloud coverage dataset. All competing models were at least 26\% worse in terms of continuous ranked probability score compared to the best model \eqref{eq:our_mod2}.

For the best predictive model in (\ref{eq:our_mod2}), we provide posterior summaries in Table \ref{tab:post_sum2}. Additionally, we display the correlation contour as function of spherical distance $\theta$ and time $t$ for the posterior mean in Figure \ref{fig:cloud_contour}. The scale of the plots in Figure \ref{fig:cloud_contour} are not the same as Figure \ref{fig:air_contour}. Correlation falls off sharply as a function of great circle distance. In this way, the total cloud coverage dataset differs greatly from the near-surface temperature dataset which demonstrated very persistent autocorrelation over space and time.
\begin{table}[H]
\centering
\begin{tabular}{rrrrr}
  \hline
  & Mean & Std. Err. &  2.5\%  & 97.5\%  \\ 
  \hline
    $\tau^2$  & 22.280 & 3.601 & 16.764 & 31.598 \\ 
 $\sigma^2$  & 595.93 & 7.677 & 580.318 & 610.792 \\ 
 $c_s$ &  0.102 & 0.004 & 0.094 & 0.110 \\
 $c_t$  & 6.762 & 1.794 & 3.416 & 9.803 \\  
$\alpha$  & 0.350 & 0.071 & 0.232 & 0.516 \\ 
$\beta$  & 0.952 & 0.049 & 0.817 & 0.999 \\ 
   \hline
\end{tabular}
\caption{Posterior summaries for the total cloud coverage dataset for parameters for the model fit using (\ref{eq:our_mod2}). Percentiles (2.5\% and 97.5\%) represent posterior percentiles.}\label{tab:post_sum2}
\end{table}
\begin{figure}[H]
\begin{center}
\includegraphics[width=.48\textwidth]{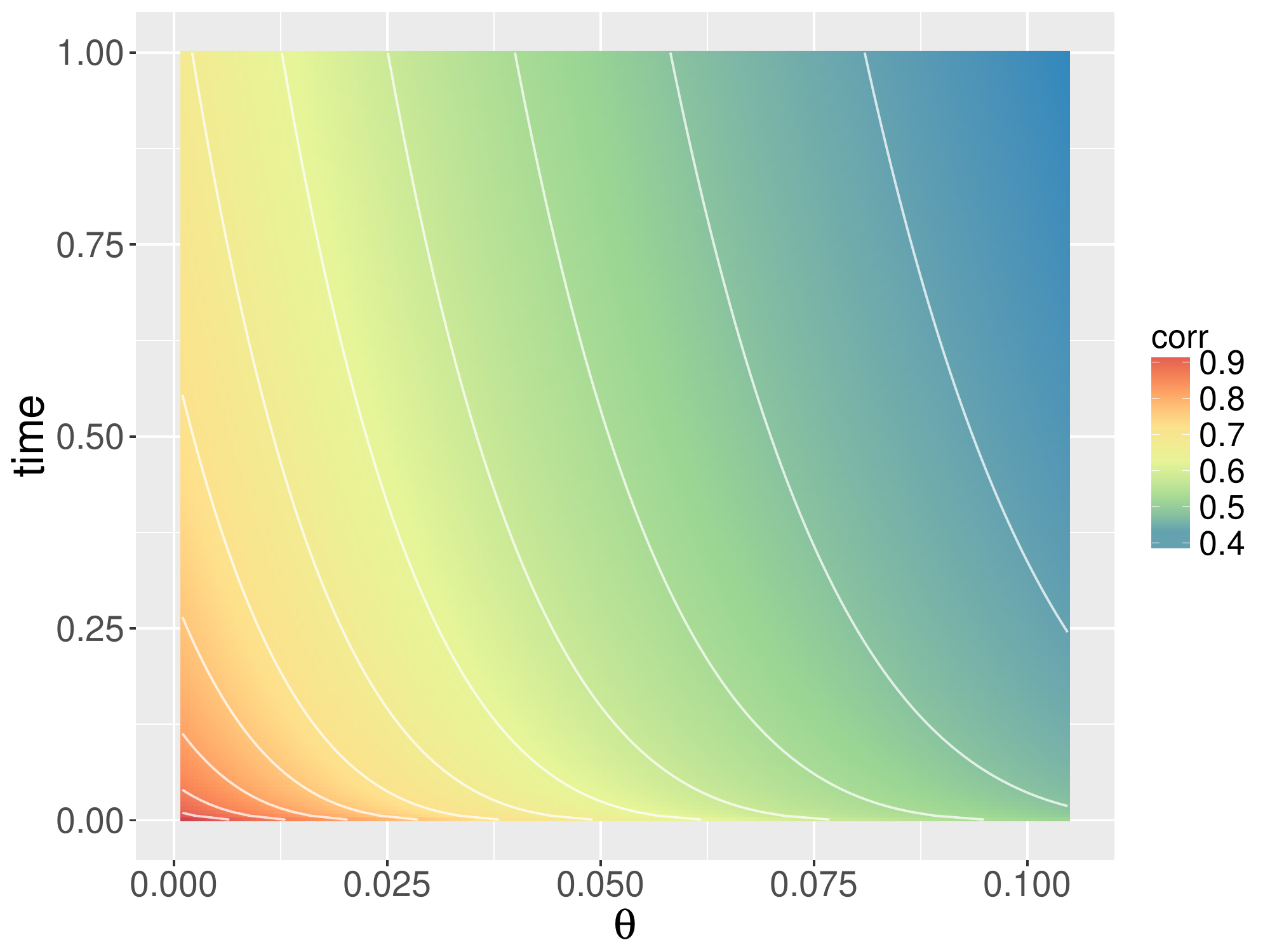}
\vspace{-5mm}
\caption{Posterior mean correlation contour plot for total cloud coverage data.}\label{fig:cloud_contour}
\end{center}
\end{figure}

For the total cloud coverage data, spatiotemporal variance $\sigma^2$ accounts for 96.40\% of the total variance $\sigma^2 + \tau^2$ (see Table \ref{tab:post_sum2}). In Table \ref{tab:post_sum2}, the parameter $c_t$ suggests that the surface temperature process exhibits persistent correlation over time; however, as discussed, the parameter $c_s$ indicates rapid decay as a function of space. The separability parameter $\beta \in [0,1]$ is close to one, meaning that the covariance process is nonseparable.

\section{Discussion and Conclusion}\label{sec:conc}

In this paper, we generalize the Gneiting criteria for nonseparable covariance functions \citep{gneiting2002} in Theorem \ref{partial-gneiting} and present new classes of nonseparable covariance models for spatiotemporally data in Theorems \ref{thm2} and \ref{heine}. In a simulation study, we explored the identifiability of covariance parameters for two covariance functions from Theorem \ref{thm2} and noted that these covariance functions have limited parameter identifiability. However, some of these challenges are remedied by fixing the spatial range parameter. We then illustrate the utility of our new Gneiting-like class using spherical distance through two climate reanalysis datasets from the National Centers for Environmental Prediction and National Center for Atmospheric Research \citep{kalnay1996}. In these two data examples, covariance models from Theorem \ref{thm2} outperform similar stationary nonseparable covariance models from \citet{gneiting2002} and \citet{Porcu-Bevilacqua-Genton} using continuous ranked probability scores, root mean squared error, and mean absolute error. As discussed, we do not suggest that covariance functions from our new covariance classes are preferable for all datasets. However, these results highlight the benefit of allowing spatial distance to be scaled by the difference in time and the importance of using the spherical distance relative to Euclidean distance or chordal distance for these datasets.

The result in Theorem \ref{partial-gneiting} presents a key for extending results obtained in Euclidean spaces to spheres cross time. Due to the lack of literature for multivariate cross-covariance models on spheres over time, with the notable exception of \cite{alegria}, we recommend this as a valuable area of expansion. In addition, the development of nonstationary covariance models for spheres cross time is an important direction for future research. 

\appendix

\section{Proofs for Theorems}\label{app:proof_all} \label{C}
\textcolor{white}{ }

\subsection*{Proof of Theorem \ref{partial-gneiting}}

 We start by proving the implication $(1) \longrightarrow (2)$. Let $C \in \Psi_{d,T}$.   Then, according to Theorem 3.3 in \cite{berg-porcu}, $C$ admits the expansion
 \begin{equation}
 \label{A2.1} C(\theta,u) = \sum_{k=0}^{\infty} b_{k,d}(u) {\cal G}_{k}^{(d-1)/2} (\cos \theta), \qquad (\theta,u) \in [0,\pi] \times \R,
 \end{equation} with $b_{k,d}$ being positive definite on $\R$ for all $k=0,1,\ldots$, and with $\sum_{k=0}^{\infty} b_{k,d}(0) < \infty$. 
 Let us start by noting that the assumption $\sum_k \int \big|  b_{k,d}(u)\big|  {\rm d}u < \infty $  
 implies $C(\theta,\cdot ) \in L_1(\R )$ for all $\theta \in [0,\pi]$. In fact,
\begin{eqnarray*}
 \int_{-\infty}^{\infty} \mid  C(\theta,u) \mid  {\rm d}u&=& \int_{-\infty}^{\infty} \bigg |  \sum_{k=0}^{\infty} b_{k,d}(u) {\cal G}_{k}^{(d-1)/2}(\cos \theta)\bigg | {\rm d} u   \\ & \le & \sum_{k=0}^{\infty} \int_{-\infty}^{\infty} \bigg | b_{k,d}(u) \bigg |  {\rm d} u < \infty, 
\end{eqnarray*}
where the last step is justified by the fact that, for normalized Gegenbauer polynomials, $\big|   {\cal G}^\lambda_k(u) \big|  \le {\cal G}_k^\lambda(1)$, $\lambda > 0$.  

Let $C_{\tau}$ be the function defined through (\ref{partial-gneiting2}). Since $C(\theta,\cdot) \in L_1(\R)$ for all $\theta$, and using Lebsegue's theorem, we have 
\begin{eqnarray*}
C_{\tau}(\theta) &=&  \int_{-\infty}^{+\infty} {\rm e}^{- \mathsf{i} u \tau} C(\theta,u) {\rm d} u \\
&=&   \int_{-\infty}^{+\infty} {\rm e}^{- \mathsf{i} u \tau}  \sum_{k=0}^{\infty} b_{k,d}(u) {\cal G}_{k}^{(d-1)/2} (\cos \theta) {\rm d} u \\
&=&  \sum_{k=0}^{\infty} \widehat{b}_{k,d}(\tau) {\cal G}_{k}^{(d-1)/2} (\cos \theta), \qquad \theta \in [0,\pi],
\end{eqnarray*}
where $\widehat{b}_{k,d}(\tau)= \int_{-\infty}^{+\infty} {\rm e}^{- \mathsf{i} u \tau}b_{k,d}(u) {\rm d} u$.
Clearly, for any $k=0,1, \ldots$ we have that $\widehat{b}_{k,d}$ is nonnegative and additionally $\widehat{b}_{k,d} \in L_1(\R)$. To complete the proof, we invoke the theorem of \cite{schoenberg} and thus need to show that $\sum_k \widehat{b}_{k,d}(\tau) < \infty$ for all $\tau \in \R$. Again invoking Lebsegue's theorem we have
\begin{eqnarray*}
 \sum_{k=0}^{\infty}  \widehat{b}_{k,d}(\tau) &=& \sum_{k=0}^{\infty} \int_{-\infty}^{+\infty} {\rm e}^{- \mathsf{i} u \tau}b_{k,d}(u) {\rm d} u \\ &=& \int_{-\infty}^{+\infty} {\rm e}^{- \mathsf{i} u \tau} \sum_{k=0}^{\infty} b_{k,d}(u) {\rm d} u   \\
&=&  \int_{-\infty}^{+\infty} {\rm e}^{- \mathsf{i} u \tau} B_d(u) {\rm d}u,
 \end{eqnarray*}
with $B_{d}(u)= \sum_{k=0}^{\infty} b_{k,d}(u)$. Using the fact that $\sum_{k=0}^{\infty} b_{k,d}(0)<\infty$ and that $b_{k,d}(0) \geq \mid b_{k,d}(u)\mid $ for all $u \in \R$ (because $b_{k,d}$ are positive definite for all $k$), we get
$$ \infty > \sum_{k=0}^{\infty} b_{k,d}(0) \ge \sum_{k=0}^{\infty} \mid b_{k,d}(u) \mid  \ge \sum_{k=0}^{\infty}b_{k,d}(u) = B_d(u), \quad u \in \R, $$
showing that $B_d$ is bounded and continuous. Further, $B_d$ is positive definite on $\R$ because positive definite functions are a convex cone being closed under pointwise convergence. To complete the result, we need to prove that $B_d \in L_1(\R)$. This comes from the fact that 
$$ B_d(u) = C(0,u) \in L_1(\R). $$
The proof is completed.  

To prove $(2) \longrightarrow (1)$, we let $C_{\tau}$ as defined through (\ref{partial-gneiting2}) and suppose that $C_{\tau} \in \Psi_d$ a.e. $\tau \in \R $. By \cite{schoenberg} theorem, we have that
\begin{equation}
\label{d-schoen} \widetilde{b}_{k,d}(\tau): = \kappa \int_{0}^{\pi} C_{\tau}(\theta) {\cal G}_k^{(d-1)/2}(\cos \theta) \sin \theta^{d- 1} {\rm d} \theta, \qquad \tau \in \R, 
\end{equation}is nonnegative, where $\kappa>0$ \citep{berg-porcu}. Using again \cite{schoenberg} theorem, we can write $C_{\tau}$ as
$$ 
C_{\tau}(\theta) = \sum_{k=0}^{\infty} \widetilde{b}_{k,d}(\tau) {\cal G}_k (\cos \theta), \qquad \theta \in [0,\pi].
$$
Since $C_{\tau} \in \Psi_{d}$ a.e. $\tau$, this in turn implies $\infty >C_{\tau}(0) = \sum_k \widetilde{b}_{k,d}(\tau)$ for all $\tau \in \R$. We now define 
$$ \widetilde{B}_d(\tau):= \sum_{k=0}^{\infty} \widetilde{B}_{k,d}(\tau), \qquad \tau \in \R. $$
Apparently $\widetilde{B}_d$ is nonnegative. Let us now show that $\widetilde{B}_d \in L_1(\R)$. To do so, we note that 
\begin{eqnarray*}
\int_{-\infty}^{+\infty} \Big |  \widetilde{B}_d(u) \Big |{\rm d} u & = & \int_{-\infty}^{+\infty} \Big |  \sum_{k=0}^{\infty} \widetilde{b}_{k,d}(u) \Big | {\rm d} u \\
& \le & \sum_{k=0}^{\infty} \int_{-\infty}^{+\infty} \Big |  \widetilde{b}_{k,d}(u) \Big | {\rm d} u < \infty,
\end{eqnarray*}
so that $\widetilde{B}_d \in L_1(\R)$ as asserted. This in turn implies that $b_{k,d} \in L_1(\R)$ for all $k = 0,1,\ldots$. Thus, we can define a function $C:[0,\pi] \times \R \to \R$ through
\begin{eqnarray*}
 C(\theta,u) &=& \frac{1}{2\pi} \int_{-\infty}^{+\infty} {\rm e}^{- \mathsf{i} u \tau} C_{\tau}(\theta) {\rm d} \tau \frac{1}{2\pi} \int_{-\infty}^{+\infty} {\rm e}^{- \mathsf{i} u \tau} \sum_{k=0}^{\infty } \widetilde{b}_{k,d}(\tau) {\cal G}_{k} (\cos \theta)  \\ 
&=&   \sum_{k=0}^{\infty} b_{k,d}(u) {\cal G}_{k} (\cos \theta), \qquad \theta\in [0,\pi], u\in \R,
\end{eqnarray*}

 and where $b_{k,d}(\cdot)=1/(2 \pi) \int {\rm e}^{\mathsf{i} \cdot \tau} \widetilde{b}_{k,d}(\tau) {\rm d } \tau$ is positive definite on $\R$ for all $n \in \mathbb{N}$. 
 Thus, the proof is completed by invoking Theorem 3.3 in \cite{berg-porcu} and by verifying that $$ \sum_{k=0}^{\infty} b_{k,d}(0) = \sum_{k=0}^{\infty} \int_{-\infty}^{+\infty} \widetilde{b}_{k,d}(\tau) {\rm d} \tau < \infty.$$ 
We now prove the implication $(2) \longrightarrow (3)$. Since $C_{\tau} \in \Psi_d$ for almost every $\tau \in \R$, we have that (\ref{d-schoen}) holds. This implies that $$ b_{k,d}(u) = \frac{1}{\pi} \int_{-\infty}^{+\infty} {\rm e}^{\mathsf{i} u \tau } \widetilde{b}_{k,d}(\tau) {\rm d} \tau, \qquad  k=0,1,\ldots, $$
is positive definite. Summability of the sequence $\{b_{k,d}(u) \}_{k=0}^{\infty}$ at $u=0$ follows easily from previous arguments.  \\
To prove the implication $(3) \longrightarrow (2)$, using (\ref{sch-coeff}) we have 
\begin{eqnarray*}
b_{k,d}(u) &=& \int_{0}^{\pi} C(\theta,u) {\cal G}_{k}^{(d-1)/2} (\cos \theta) \sin \theta^{d-1} {\rm d} \theta \\
&=& \frac{1}{2\pi} \int_{-\infty}^{+\infty} {\rm e}^{\mathsf{i}u \tau}  \int_{0}^{\pi} C_{\tau}(\theta) {\cal G}_{k}(\theta) \sin \theta^{d-1} {\rm d} \theta {\rm d} \tau,
\end{eqnarray*}
which shows that $C_{\tau} \in \Psi_d$ for all $\tau$ because the positive definiteness of $b_{k,d}(\cdot)$ implies, by Lemma 4.3 in \cite{berg-porcu}, the inner integral $\int_{0}^{\pi} C_{\tau}(\theta) {\cal G}_{k}(\theta) \sin \theta^{d-1} {\rm d} \theta$ to be nonnegative. \\
To conclude, the implication $(3 ) \longrightarrow (1)$ has been shown by \cite{berg-porcu}. \hfill $\Box$

\subsection*{Proof of Theorem \ref{thm2}}

We start by noting the beautiful formula
$$ \int_{0}^{\infty} {\rm e}^{-r x } {\rm e}^{- r \xi} {\rm d} r = \frac{1}{x+\xi}, \qquad \xi>0, \; x \ge 0. $$
 We now consider the function
$$ H_{\xi}(\theta,u) = \int_{0}^{\infty} {\rm e}^{-r \theta } {\rm e}^{- r \xi \psi(u^2)} {\rm d} r = \frac{1}{\psi(u^2)} \left ( \xi +\frac{\theta}{\psi(u^2)}\right )^{-1} , \qquad (\theta,u) \in [0,\pi] \times \R, $$
which shows that $H_\xi$ is positive definite on every $d$-dimensional sphere cross time ($\R$)
because the mapping $\theta \mapsto \exp(- r \theta)$ is positive definite on every $d$-dimensional sphere $\S^d$ \citep[][Theorem 7]{gneiting2013} and because $u \mapsto \exp(- \xi r \psi(u^2))$ is positive definite on the real line. Since $\varphi \in {\cal S}$, we have, using  (\ref{stieltjes}), that
$$ C(\theta,u) =\frac{1}{\psi(u^2)} \varphi \left (  \frac{\theta}{\psi(u^2)} \right )  =  \int_{0}^{\infty} H_{\xi}(\theta,u) \mu({\rm d} \xi), $$
and this proves the assertion.  \hfill $\Box$

\subsection*{Proof of Theorem \ref{heine}}

We make use of the arguments in the proof of Theorem \ref{thm2}, in concert with formula (15) on page 15 of \cite{bateman}:
\begin{eqnarray*}
&& \frac{\pi}{2 } \int_{0}^{\infty}  \cos(t \omega) \exp \left ( -(1+\omega^2) x \right ) \frac{{\rm d} \omega}{1 + \omega^2} \\&=& {\rm e}^{-u} {\rm erfc} \left ( \sqrt{x} - \frac{t}{2 \sqrt{x} }\right ) +  {\rm e}^{u} {\rm erfc} \left ( \sqrt{x} + \frac{t}{2 \sqrt{x} }\right ), 
\end{eqnarray*}
for $x,t \ge 0$. We now replace $x$ with $\psi_{[0,\pi]}(\theta)$ and $t$ with $|u| $. Since the composition of the negative exponential with a positive functions having completely monotonic derivative provides a completely monotonic function, we can invoke Theorem 7 in \cite{gneiting2013} to infer that the mixture above provides, in view of analogous arguments to the proof of Theorem \ref{thm2}, a positive definite function on $\S^d \times \R$ for all $d$. The proof is completed. \hfill $\Box$

\subsection*{Supplementary Theorem for Section \ref{Sec3a}}\label{app:extra_thm}

We now show how Theorem \ref{partial-gneiting} can be useful to understand connections and analogues between the class of positive definite functions $\R \times \R$ and positive definite functions on the circle $\S^1$ cross $\R$.

\begin{theorem} \label{alan-gelfand}
Let $\varphi: \R \times \R \to \R$ be a covariance function that is symmetric in both arguments. Let $\varphi_{\tau}:\R \to \R$ be the function defined by
\begin{equation} \label{ep}
\varphi_{\tau}(x) = \int_{-\infty}^{+ \infty } {\rm e}^{\mathsf{i} u \tau } \varphi(x,u) {\rm d} u, \qquad x \in \R.
\end{equation} and suppose that such an integral is well defined. 
Let $\varphi_{\tau}(x)=0$ whenever $| x|  \ge \pi$. Call $C(\theta,u)= \varphi_{[0,\pi]}(\theta,u)$, $\theta \in [0,\pi]$, $u \in \R$, where the restriction to $[0,\pi]$ is with respect to the first argument. Then, $C(\theta,u)$ is a covariance function on $\S^3 \times \R$.
\end{theorem}

\begin{proof}
Since $\varphi$ is positive definite in $\R \times \R$, by Lemma 1 in \cite{gneiting2002} we get that $\varphi_{\tau}$ is positive definite in $\R$ a.e. $\tau \in \R$. Additionally, $\varphi_{\tau}(x)=0$ whenever $|x|\ge \pi$. Call $\psi_{\tau}$ the restriction of $\varphi_{\tau}$ to $[0,\pi]$. By Corollary 3 in \cite{gneiting2013} we have that the coefficients $\widetilde{b}_{k,1}(\tau)$, $\tau \in \R$ in the Schoenberg expansion of $\psi_{\tau}$, as defined in (\ref{d-schoen}) are nonnegative and strictly decreasing in $k$ for any fixed $\tau \in \R$. This implies that $\varphi_{\tau}(\theta)$
is positive definite in $\S^3$
 for almost every $\tau \in \R$. Application of Theorem \ref{partial-gneiting}, Assertion 2, shows that $C(\theta,u)$ is positive definite in $\S^3 \times \R $. The proof is completed. 
\end{proof}

\section{Modeling Details for the Nearest Neighbor Gaussian Process}\label{app:NNGP}

Suppose we begin with a parent Gaussian process over $\mathbb{R}^d \times \mathbb{R}$ or $\S^{d-1} \times \mathbb{R}$. Nearest neighbor Gaussian processes induce sparsity in the precision matrix of the parent Gaussian process by assuming conditional independence given neighborhood sets \citep{datta2016a,datta2016c}. Let $\mathcal{S} = \{(\bs_1,t_1), (\bs_2,t_2), ..., (\bs_k,t_k) \}$ of $k$ distinct location-time pairs denote the reference set, where we allow time to act as a natural ordering and impose an ordering on the locations observed at identical times. Then, we define neighborhood sets $N_\mathcal{S} = \{N(\bs_i,t_i); i = 1,...,k \}$ over the reference set with $N(\bs_i,t_i)$ consisting of the $m$ nearest neighbors of $(\bs_i,t_i)$, selected from $\{(\bs_1,t_1),(\bs_2,t_2),...,(\bs_{i-1},t_{i-1}) \}$. If $i \leq m+1$, $N(\bs_i,t_i) = \{(\bs_1,t_1), (\bs_2,t_2), ..., (\bs_{i-1},t_{i-1} ) \}$. For the Gibbs sampler, we need to define an inverse of the neighborhood set, which we call $U(\bs_i,t_i)$. The set $U(\bs_i,t_i)$ consists of all sites that include $(\bs_i,t_i)$ in their neighborhood sets. 

Along with $\mathcal{S}$, $N_\mathcal{S}$ defines a Gaussian directed acyclic graph $\bw_\mathcal{S}$ with a joint distribution
\begin{equation}
p( \bw_\mathcal{S} ) = \prod^k_{i=1} p(w(\bs_i,t_i)\mid \bw_{N(\bs_i,t_i)})
= \prod^k_{i=1} \mathcal{N}(\bw(\bs_i,t_i)\mid  \bB_{(\bs_i,t_i)} \bw_{N(\bs_i,t_i)}, \bF_{(\bs_i,t_i)}), 
\end{equation}  
where $\mathcal{N}$ is a normal distribution,
\begin{align*}
\bB_{(\bs_i,t_i)} &= C_{(\bs_i,t_i),N(\bs_i,t_i)} C_{N(\bs_i,t_i)}^{-1}, \\
\bF_{(\bs_i,t_i)} &= \sigma^2 - \bB_{(\bs_i,t_i)} C_{N(\bs_i,t_i),(\bs_i,t_i)}, 
\end{align*}
$C_{(\bs_i,t_i),N(\bs_i,t_i)}$ is a vector of covariances between $(\bs_i,t_i)$ and its neighbors, $C_{N(\bs_i,t_i)}$ is the covariance matrix for the neighbors of $(\bs_i,t_i)$, and $\bw_{N(\bs_i,t_i)}$ is the subset of $\bw_\mathcal{S}$ corresponding to neighbors $N(\bs_i,t_i)$ \citep{datta2016a}. \citet{datta2016a} extend this Gaussian directed acyclic graph to a GP. This Gaussian process formulation only requires storage of $k$ $m \times m$ distance matrices and requires many fewer floating point operations than full Gaussian process models \citep[see][]{datta2016a}. Like any other GP model, the NNGP can be utilized hierarchically for spatiotemporal random effects. In this article, we use NNGP as an alternative to the full Gaussian process specification.

We envision our model taking the following form:
\begin{align}
Y(\bs,t) &= \bx(\bs,t)^\top \bbeta + w(\bs,t) + \epsilon(\bs,t), \\
w(\bs,t) &\sim NNGP(0,C((\bs,t),(\bs',t'))),  \nonumber \\
\epsilon(\bs,t) &\sim GP(0,\tau^2 \delta^\bs_\bs \delta^t_t),
 \nonumber
\end{align} 
where $Y(\bs,t)$ is a spatiotemporal process measured (with error), $\bx(\bs,t)$ are $p$ spatiotemporal covariates, and $\delta^b_a$ is the Kronecker delta function. We define $C((\bs,t),(\bs',t'))$ using a covariance model discussed in Section \ref{Sec2} or Section \ref{sec:theory}. We recommend using inverse gamma (IG) prior distributions for $\tau^2$ from the pure error term and $\sigma^2$ from the covariance function because this selection gives closed form full conditional distributions. If the outcomes and covariates are centered, then an intercept is unnecessary. If covariates are not available, then $\bx(\bs,t)^\top \bbeta$ is replaced with $\mu$. Here, for more compact notation, we index space-time location pairs with $i$ as $(\bs_i,t_i)$ and refer to corresponding outcomes, covariates, and spatiotemporal random effects as $y_i$, $\bx_i$, and $w_i$, respectively.

The prior mean and variance for regression coefficients $\bbeta$ are $m_\beta$ and $V_\beta^{-1}$, respectively. Additionally, let $a_V$ and $a_\tau$ be shape parameters for the inverse gamma prior distributions for $\sigma^2$ and $\tau^2$. Similarly, let $b_V$ and $b_\tau$ be scale parameters (corresponding to rate parameter of the gamma distribution) for the inverse gamma prior distributions for $\sigma^2$ and $\tau^2$. 

The full conditional distributions for the Gibbs sampler, which we denote $\cdot \mid  \cdots$, are
\begin{align*}
\bbeta \mid  \cdots &\sim \mathcal{N}_p( V_\beta^* m_\beta^*,V_\beta^*) \\
\tau^2 \mid  \cdots &\sim IG(a_{\tau}^*,b_{\tau}^*) \\
\sigma^2\mid  \cdots &\sim IG(a_V^*,b_V^*) \\
\bw_i \mid  \cdots &\sim \mathcal{N}_1( V_{w_i}^* m_{w_i}^*,V_{w_i}^*).
\end{align*} 
To express $V_{w_i}^*$ and $m_{w_i}^*$, we must define some additional terms. First, we let $B_{(\bs',t'),(\bs_i,t_i)} $ be the scalar in $B_{(\bs',t')}$ corresponding to $(\bs_i,t_i)$. Second, we define $$a_{(\bs',t'),(\bs_i,t_i)} = w(\bs',t') - \sum_{(\bs_j,t_j) \in N(\bs',t'),(\bs_j,t_j) \neq (\bs_j,t_j) } B_{(\bs',t'),(\bs_i,t_i) } w(\bs_i,t_i). $$ For more details, see \citet{datta2016a}. 

The parameters of the full conditional distributions are as follows:
{ 
\begin{align*}
    V_\beta^* &= \left( \bX^\top\bX /\tau^2 + V_\beta^{-1} \right)^{-1} \\
    m_\beta^* &= V_\beta^{-1} m_\beta + \sum_i\bx_{i} ( y_{i} - w_i) / \tau^2  \\
    V_{w_i}^* &= \left(1/\tau^2 + F_{(\bs_i,t_i)}^{-1} + \sum_{(\bs',t'): (\bs',t') \in U(\bs_i,t_i)}   B_{(\bs',t'),(\bs_i,t_i)}^2 / F_{(\bs',t')}   \right)^{-1}\\
  m_{w_i}^* &= (y_{i} - \bx_{i}^\top \bbeta ) /\tau^2 +  B_{(\bs_i,t_i)} \bw_{N(\bs_i,t_i)} / F_{(\bs_i,t_i)} + \\ & \sum_{(s',t'): (s',t') \in U(\bs_i,t_i)}  B_{(\bs',t'),(\bs_i,t_i)} a_{(\bs',t'),(\bs_i,t_i)} /  F_{(\bs',t')}  \\
a_V^* &= a_V + n/2 \\
b_V^* &= b_V + \sigma^2 \sum_i  (w_i - B_{(\bs_i,t_i)} {\bw_{N(\bs_i,t_i)}} )^\top (w_{i} - B_{(\bs_i,t_i)} {\bw_{N((\bs_i,t_i))}}) /F_{(\bs_i,t_i)} \\
a_{\tau}^* &= a_\tau + \frac{n}{2} \\
b_{\tau}^* &= b_\tau + \frac{1}{2} \sum_i  (y_{i} - \bx_{i}^\top \bbeta - w_i)^2 .  \\
\end{align*}
}

Prediction at an arbitrary location and time requires selection of $m$ nearest neighbors from the reference set $\mathcal{S}$ for that location-time pair. We discuss two ways of selecting $m$ neighbors. In theory, any location-time pair from the reference set can be selected as a neighbor for any prediction. If we allow predictions to depend on data occurring after the prediction time, then we call this a \emph{retrospective} prediction since such a prediction could only be made retrospectively. On the other hand, a \emph{prospective} prediction limits neighbor selection to elements of $\mathcal{S}$ that occur at the same time or prior to the time of prediction. \citet{datta2016c} selects neighbors for prospective predictions, and we do the same in our analyses.

Predicted spatiotemporal random effects at location-time pairs follow a conditional normal distribution, where conditioning is limited to its neighbors. For any space-time pair $(\bs,t)$, the conditional distribution of the random effect is 
\begin{equation}
\bw(\bs) \vert \bw_{N(\bs)} \sim \mathcal{N}\left( C_{\bs,N(\bs)} C_{N(\bs)}^{-1} \bw_{N(\bs)} , \\
C(\bs,\bs) - C_{\bs,N(\bs)} C_{N(\bs)}^{-1} C_{\bs,N(\bs)}^\top \right). \label{eq:cond}
\end{equation}
Then, the posterior prediction $Y(\bs,t) \mid  \bY$ is $\bx(\bs,t)^\top \bbeta + w(\bs,t) + \epsilon_i(\bs,t)$, where, in practice, posterior samples of $\bbeta$, $w(\bs,t)$, and $\tau^2$ are used to sample from the posterior predictive distribution. Predictions at hold-out location-time pairs can be used to compare competing models. 

\section{Locations of Hold-out Data from Section \ref{sec:Data}}\label{app:holdout}

Here, we provide the locations of hold-out data used for model validation. The locations for hold-out air temperature are given in Figure \ref{fig:air_holdout}. The locations for hold-out cloud coverage are in Figure \ref{fig:cloud_holdout}.

\vspace{-20mm}
\begin{figure}[H]
\begin{center}
\includegraphics[width=.8\textwidth]{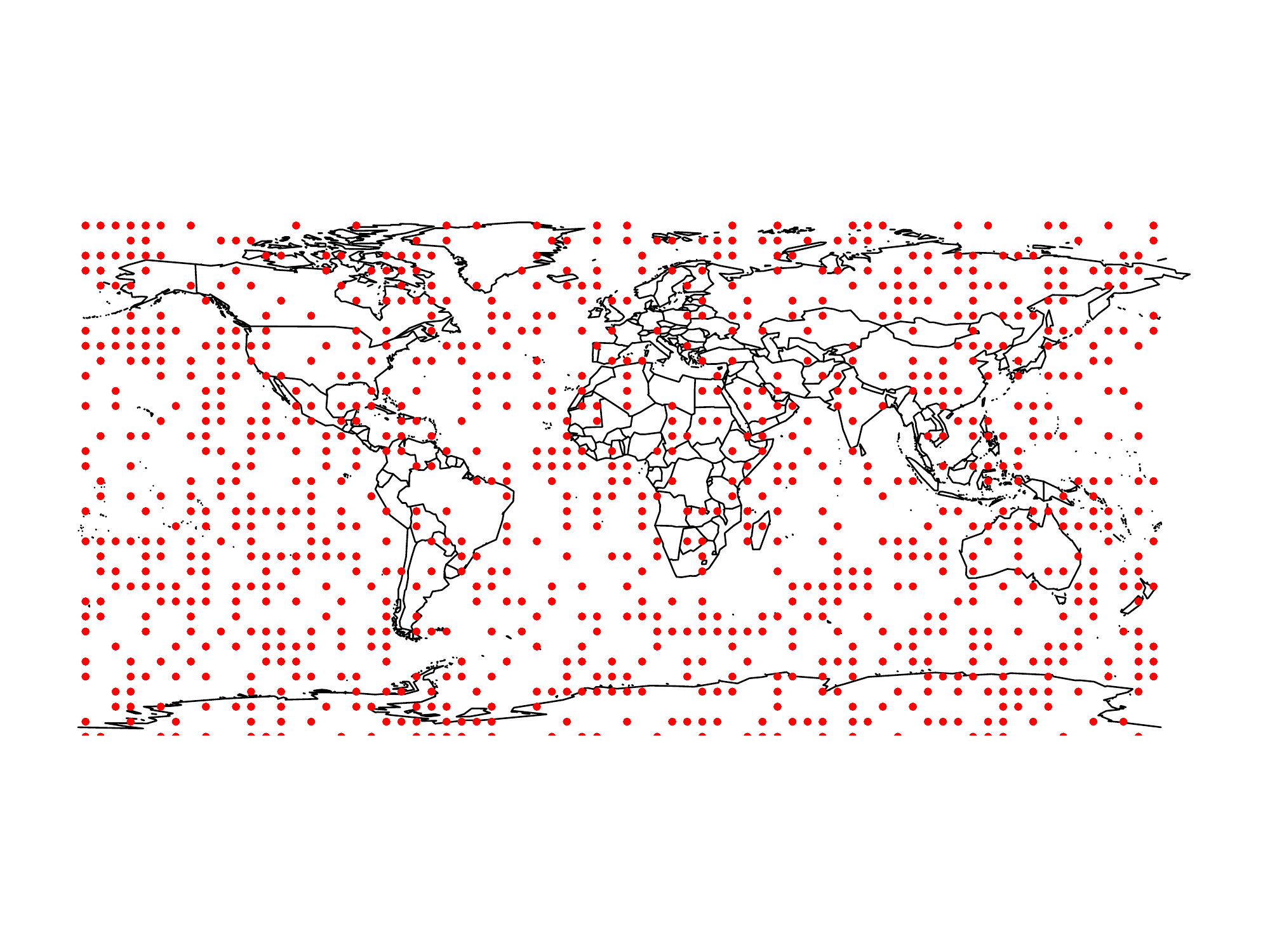}
\vspace{-20mm}
\caption{Hold-out locations used for predictive performance for the near-surface air temperature.}\label{fig:air_holdout}
\end{center}
\end{figure}
\vspace{-20mm}
\begin{figure}[H]
\begin{center}
\includegraphics[width=.8\textwidth]{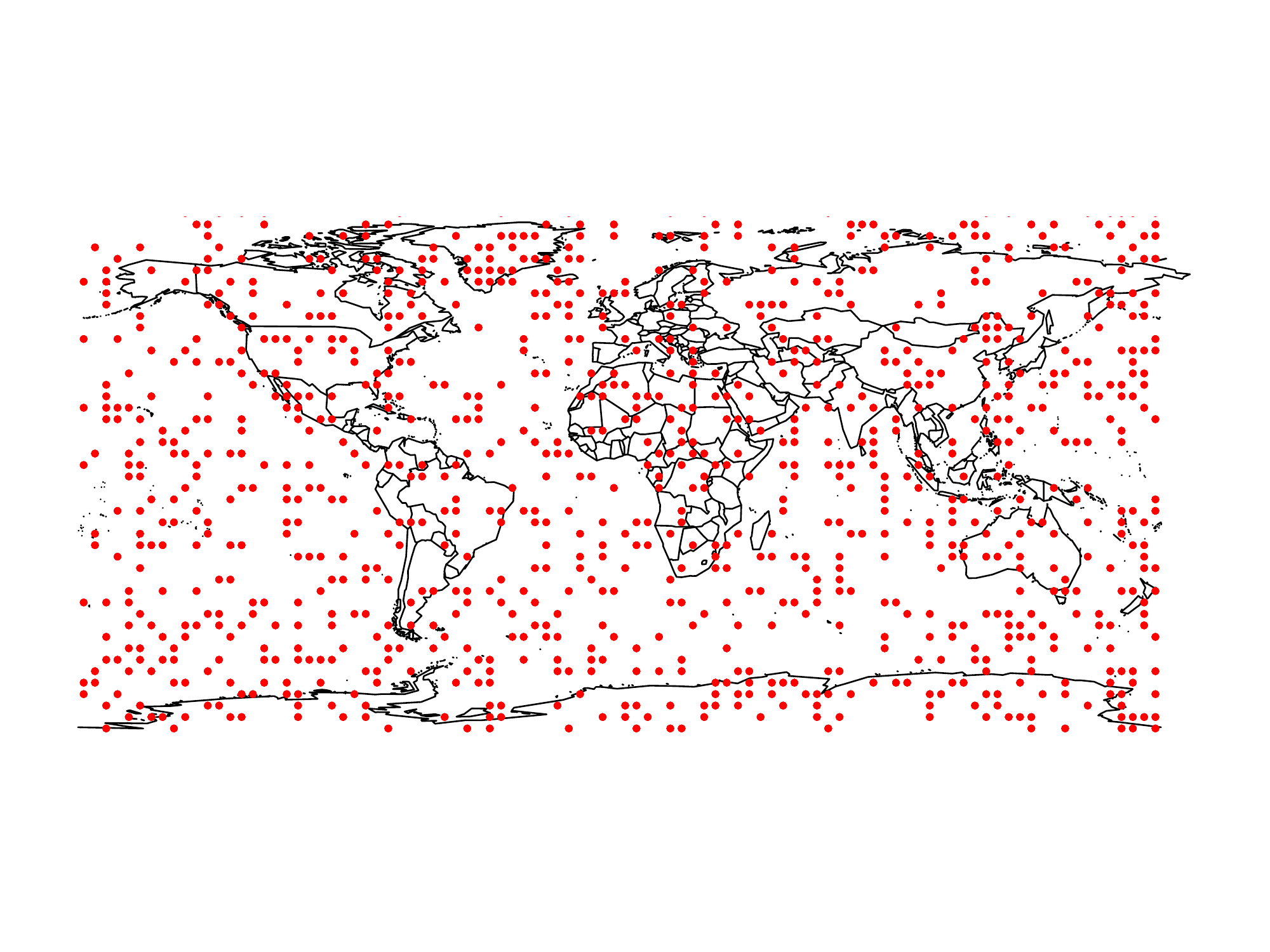}
\vspace{-20mm}
\caption{Hold-out locations used for predictive performance for the total cloud coverage dataset.}\label{fig:cloud_holdout}
\end{center}
\end{figure}

\bibliographystyle{apalike}
\bibliography{refs}

\end{document}